\documentclass[12pt, a4paper]{article}

\usepackage[utf8]{inputenc}
\usepackage[T1]{fontenc}

\usepackage{amsmath, amssymb, amsthm}
\usepackage{mathtools}

\usepackage[top=2.5cm, bottom=2.5cm, left=3cm, right=3cm]{geometry}
\usepackage{setspace}
\usepackage{enumitem}
\usepackage{booktabs}
\usepackage{multirow}
\usepackage{graphicx}
\usepackage{caption}
\usepackage{subcaption}
\usepackage{float}

\usepackage{natbib}
\usepackage[colorlinks=true, linkcolor=blue, citecolor=blue, urlcolor=blue]{hyperref}

\usepackage{algorithm}
\usepackage{algorithmic}

\usepackage{listings}
\usepackage{xcolor}

\newtheorem{theorem}{Theorem}[section]
\newtheorem{proposition}[theorem]{Proposition}

\newtheorem{definition}{Definition}[section]
\newtheorem{remark}{Remark}[section]

\newcommand{\R}{\mathbb{R}}
\newcommand{\E}{\mathbb{E}}
\newcommand{\Prob}{\mathbb{P}}
\newcommand{\Q}{\mathbb{Q}}
\newcommand{\F}{\mathcal{F}}

\newcommand{\dt}{\,\mathrm{d}t}

\title{\textbf{Neural Networks Learning the Radon--Nikodym Derivative:\\
Empirical Option Pricing in Incomplete Markets}}

\author{Ziyuan Zhang, Kiseop Lee\\
\small Department of Statistics \\
\small Purdue University}

\date{\today \quad (Draft, please do not cite)}

\begin{document}

\maketitle

\begin{abstract}
In incomplete markets, the no-arbitrage condition (NFLVR) guarantees only the existence, not the
uniqueness, of an equivalent local martingale measure (ELMM): unhedgeable risk sources such as jump
risk and stochastic volatility generate an entire family of equivalent measures satisfying the
condition, different measures correspond to different option prices, and the measure the market
actually selects cannot be uniquely determined by the asset dynamics themselves. This paper
systematically characterizes the \emph{identifiability} of $\Q$ from option data, and on that basis
proposes a pricing measure that can be identified from data and that prices \emph{any} claim
consistently. Our starting point is an identification boundary (an ``identification wall''): European
options identify only the terminal marginal distribution; out-of-sample tails and path/joint structure
cannot be determined from observed call prices, and require additional instruments \emph{matched} to the
priced risk (variance and higher-moment swaps, path-dependent claims, etc.) in order to be identified.
Within this framework, the classical minimum-relative-entropy weighted Monte Carlo (WMC) is the optimal
baseline for identifying the \emph{marginal}; we generalize it into a full path-space measure change,
parameterized by a neural network and defined on physical scenarios.

Concretely, this paper proposes XiNet---a data-driven method that requires no specification of a
parametric model, learning the Radon--Nikodym derivative $\xi = \mathrm{d}\Q/\mathrm{d}\Prob$ directly
from observable path statistics, using the option-price cross-section as a soft constraint, and
assigning each Monte Carlo path a path-specific measure-change weight. To free the framework genuinely
from model assumptions, we construct, from the daily log-return series, an 8-dimensional model-free
feature set (realized volatility, maximum single-day absolute return, skewness, excess kurtosis,
volatility time-variation, volatility clustering, maximum drawdown, cumulative log-return); these
features require no parameter estimation and can be computed directly from simulated or real market
data.

In settings with an explicit risk premium (Merton jump, Heston volatility, CGMY Lévy-measure risk
premium), $\mathrm{d}\Q/\mathrm{d}\Prob$ is path-dependent. We distinguish two classes of pricing
problems and obtain a clear and honest picture. For \emph{European} (marginal) pricing: within the
calibration range, minimum-relative-entropy WMC is already near-optimal, and XiNet matches but does not
surpass it---path features bring no additional identifying power on a purely marginal problem, consistent
with the theory that ``European options identify only the terminal marginal''; and under strict
out-of-sample conditions (deep out-of-the-money entirely excluded from calibration), both are bound by
the identification wall and cannot reliably recover the tail. For \emph{path-dependent} claims: the
situation reverses. Calibrated on the \emph{exact same} European option surface, a per-marginal method
is structurally unable to price them (an at-the-money forward-start option is mispriced on the order of
$+100\%$), whereas the single self-consistent measure learned by XiNet keeps the bias to $+0.3\%$ (jump
risk premium), significantly outperforming the equally self-consistent maximum-entropy WMC ($-24\%$)---because
its parameterization $\xi=f_\theta(\text{path features})$ captures the joint structure that European
marginals cannot constrain. Identification is thus \emph{risk-specific}: the terminal marginal is
identified by European options, the tail by moment/variance instruments, and the joint and path
structure by path-dependent instruments; and XiNet is a single pricing measure that absorbs whatever
instruments are available and prices all claims (including exotics) consistently. This provides an
operational framework, with a clear identification boundary, for ``identifying the market pricing measure
in order to price and hedge'' in incomplete markets.
\end{abstract}


\tableofcontents
\newpage

\section{Introduction}

The theoretical foundation of option pricing rests on the fundamental theorems of asset pricing.
Following the seminal work of \citet{delbaen1994general}, under the condition of no free lunch with
vanishing risk (NFLVR), the absence of arbitrage in the market is equivalent to the existence of an
equivalent local martingale measure (ELMM) $\Q\sim\Prob$ under which discounted asset prices form a
local martingale~\citep{protter_partial_2001}, so that the European call price can be written as
\begin{equation}
  C(K, T) = e^{-rT}\,\E^{\Q}\!\left[\left(S_T - K\right)^+\right].
  \label{eq:pricing_formula}
\end{equation}
Market completeness further guarantees the uniqueness of the ELMM, in which case
Eq.~\eqref{eq:pricing_formula} gives the unique replication-cost price.

Real financial markets, however, are generally incomplete: unhedgeable risk sources such as jump risk
(e.g., the Merton jump-diffusion model~\citep{merton1976option}) and stochastic volatility (e.g., the
Heston model~\citep{heston1993closed}) render the ELMM non-unique, and an entire family of equivalent
measures $\mathcal{Q}$ satisfies the no-arbitrage condition. Then different $\Q\in\mathcal{Q}$
correspond to different option prices, and the no-arbitrage condition itself cannot further discriminate
``the measure the market actually selects.'' By the Radon--Nikodym theorem, any ELMM is completely
characterized by its derivative $\xi=\mathrm{d}\Q/\mathrm{d}\Prob$ with respect to $\Prob$, so the
question of ``which $\Q$ to choose'' is equivalent to ``which $\xi$ to determine''---the upstream problem
that every risk-neutral pricing method must confront but routinely skips by default.

At the computational level, when an analytical pricing formula is unavailable, Monte Carlo simulation is
the general means of computing Eq.~\eqref{eq:pricing_formula}, and finite-sample paths almost surely
violate the martingale property, causing a systematic pricing bias. The empirical martingale simulation
(EMS) proposed by \citet{duan_empirical_1998} forces the discounted sample mean to equal the current
price via a recursive multiplicative correction, recovering the martingale property in the finite-sample
sense and greatly reducing the variance. But the common premise of EMS and its subsequent extensions
(see Section~\ref{sec:related}) is that $\Q$ or $\mathrm{d}\Q/\mathrm{d}\Prob$ is known in some
parametric sense: they are essentially finite-sample corrections to an \emph{already-chosen} martingale
measure, and do not touch the measure identification problem itself---yet this default premise simply
cannot be met in incomplete markets.

This upstream problem---\emph{how to identify the ELMM from market data}---is precisely the starting
point of this paper. We propose XiNet, which uses a neural network with observable path statistics as
inputs to learn the Radon--Nikodym derivative $\xi=\mathrm{d}\Q/\mathrm{d}\Prob$ directly, using the
option-price cross-section as a soft constraint, and identifies from data the equivalent martingale
measure consistent with market pricing.

\subsection{Review of existing methods}
\label{sec:related}

Existing work on identifying the equivalent martingale measure in incomplete markets can be broadly
grouped into four categories, which approach the ELMM from different angles but all leave gaps---in
identification, path coherence, or generalization ability---that this paper tries to bridge. The first
category is \emph{parametric model calibration}: given a dynamic model (such as Heston or Bates), the
risk-neutral parameters are determined by minimizing the difference between model prices and market
prices; the pricing quality of this route depends on the model specification, and when the true dynamics
deviate from the assumption, the calibration result may exhibit a systematic bias---this paper takes
Black--Scholes implied-volatility calibration as the representative baseline of this route. The second
category is \emph{nonparametric density estimation}: by the Breeden--Litzenberger
formula~\citep{breeden1978prices}, the risk-neutral density can be written as the second derivative of
option prices with respect to the strike, $f^{\Q}(K;T) = B(0,T)^{-1}\,\partial^2 C(K,T)/\partial K^2$,
so the risk-neutral distribution can be extracted directly from the cross-sectional quotes without a
parametric model; its difficulty lies in the instability of numerical second-order differentiation of
discrete, noisy quotes---noise is amplified by second differencing, readily producing artifacts such as
negative densities or wild oscillations, which is especially pronounced on real data with sparse quotes
and wide bid--ask spreads. Although in practice this can be controlled with regularization such as
smoothing splines or penalized regression, the choice of regularization strength is subjective, and this
method gives the cross-sectional density at each maturity, making it hard to generate temporally coherent
risk-neutral paths and offering limited support for pricing path-dependent products. The third category
is \emph{scenario reweighting} (weighted Monte Carlo): \citet{avellaneda1998minimum} reweights a
pre-generated scenario set by minimizing relative entropy subject to satisfying the option-price
constraints, and \citet{elices2008conditions} further points out that if the constraints are imposed only
at the marginal level of each maturity rather than at the path level, the resulting measure need not have
the martingale property; this class of methods solves for a set of weights point by point on a fixed
scenario set, the weights are bound to the specific sample and have no ability to generalize to new
paths, and changing the scenario set requires re-solving, whereas what this paper learns is a
\emph{function} from path features to weights, which acts directly on paths unseen during training. The
fourth category is \emph{empirical martingale simulation} (EMS): \citet{duan_empirical_1998} forces the
discounted sample mean of the simulated paths to equal the current price via a recursive multiplicative
correction, recovering the martingale property in finite samples and reducing the Monte Carlo variance;
\citet{huang2014modified}, when the parametric form of $\Q$ is hard to derive explicitly (e.g., for the
GARCH family the conditional distribution is not closed under the Esscher transform), proposed empirical
P-martingale simulation (EPMS), instead imposing an empirical-martingale constraint on the
Radon--Nikodym derivative process $\Lambda_t = \mathrm{d}\Q/\mathrm{d}\Prob|_{\mathcal{F}_t}$ under
$\Prob$ to approximate the measure change; \citet{huang_multi-asset_2026} extended this framework to the
multi-asset case; and \citet{yuan_strong_2009} established strong consistency theory for the EMS
estimator; however, this class of methods is essentially finite-sample variance reduction within an
\emph{already-chosen} $\Q$, and does not touch the upstream identification problem of ``which $\Q$ to
choose,'' which is exactly the starting point of this paper. Figure~\ref{fig:concept} intuitively
compares the above routes with this paper's method: BS calibration only fits implied volatility at a few
training strikes and extrapolates linearly to out-of-range options, with an unreliable direction; Crude
MC needs the risk-neutral measure parameters to be known before it can simulate with equal weights; EMS
applies a uniform scalar scaling to $\Prob$-measure paths to satisfy the martingale constraint and cannot
characterize the heterogeneity across paths; while the XiNet proposed in this paper learns path-specific
Radon--Nikodym weights $\xi(\omega)$ via a neural network, depending on neither the $\Q$-measure
parameters nor a parametric model assumption, and in the figure the thickness and color depth of the path
lines are proportional to the magnitude of $\xi(\omega)$.

\begin{figure}[t]
  \centering
  \includegraphics[width=\textwidth]{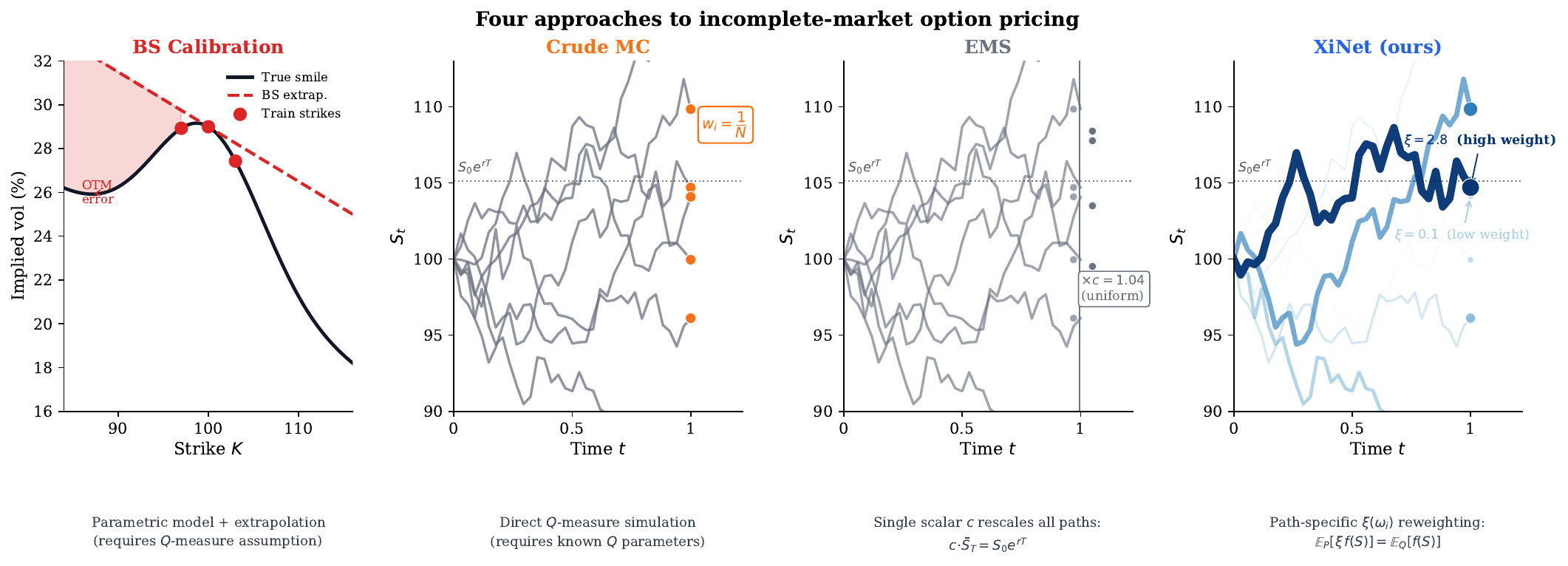}
  \caption{Conceptual comparison of four incomplete-market option pricing methods.
    \textbf{BS calibration}: fits implied volatility at a few training strikes, then extrapolates
    linearly to out-of-range options, relying on model assumptions with an unreliable extrapolation
    direction.
    \textbf{Crude MC}: equally weighted simulation under the $Q$ measure, requiring the risk-neutral
    measure parameters to be known.
    \textbf{EMS}: applies a uniform scalar scaling to $P$-measure paths to satisfy the martingale
    constraint, unable to capture the heterogeneity across paths.
    \textbf{XiNet} (this paper): learns path-specific Radon--Nikodym weights $\xi(\omega)$ via a neural
    network, depending on neither the $Q$-measure parameters nor a parametric model assumption.
    The line thickness and color depth are proportional to the magnitude of $\xi(\omega)$.}
  \label{fig:concept}

\end{figure}
\subsection{Contributions}
The core contribution of this paper is to recast the ELMM identification problem as a supervised
learning task, comprising three specific points:
\begin{enumerate}
  \item \textbf{Identification framework.} We learn the Radon--Nikodym derivative
        $\xi_i = \mathrm{d}\Q/\mathrm{d}\Prob\,(\omega_i)$ directly, assigning each $\Prob$-measure path a
        path-specific weight, and parameterize $\xi$ with a neural network so that it can capture the
        nonlinear effect of path features (especially the number of jumps) on the measure change.
  \item \textbf{Training objective.} We replace the direct relative squared price error with a
        return-matching constraint at the training strikes (the weighted expectations under $\Prob$ of the
        discounted stock and the normalized call payoff both equal $1$), and normalize the network output
        by its sample mean before pricing at inference.
  \item \textbf{Systematic evaluation.} We evaluate the method systematically in a controlled Merton
        jump-diffusion environment, and quantitatively decompose the contribution of each path feature via
        ablation experiments.
\end{enumerate}

The entire framework is fully data-driven: it requires no assumption that market data follow any
parametric distribution, only the computation of path statistics from the observable return series;
applying the learned Radon--Nikodym derivative $\xi_\theta(\phi)$ to $\Prob$-measure paths yields
risk-neutral pricing consistent with the option-price cross-section (the empirical version of
Eq.~\eqref{eq:pricing_rnd}, see Section~\ref{sec:method}).

The remainder of the paper is organized as follows:
Section~\ref{sec:theory} establishes the theoretical framework;
Section~\ref{sec:method} describes the methodology for learning $\xi$ with a neural network;
Section~\ref{sec:experiments} reports the numerical experiments, including the failure mechanisms of the
methods under risk-premium settings and the BS calibration baseline with Greek information;
Section~\ref{sec:realdata} validates marginal-level feasibility and robustness on the real SPX option
surface;
Section~\ref{sec:conclusion} concludes and outlines future work.
The appendices give a summary of the experimental parameters, a theoretical and experimental analysis of
the structural limitations of EMS, and the ablation experiments on Merton-model-specific features.

\section{Theoretical framework}
\label{sec:theory}

\subsection{Preliminaries: measure change and no-arbitrage pricing}

Let $(\Omega, \mathcal{F}, \mathbb{P})$ be a probability space, and $\mathbb{Q}\sim\mathbb{P}$ an
equivalent measure. The Radon--Nikodym theorem guarantees the existence of a unique
($\mathbb{P}$-a.s.) nonnegative random variable $\xi = \mathrm{d}\mathbb{Q}/\mathrm{d}\mathbb{P}$,
satisfying the normalization $\mathbb{E}^{\mathbb{P}}[\xi]=1$ and the change-of-measure identity
\begin{equation}
    \mathbb{E}^{\mathbb{Q}}[X] = \mathbb{E}^{\mathbb{P}}[\xi\, X],
    \qquad \text{for any } \mathbb{Q}\text{-integrable } X.
    \label{eq:cme}
\end{equation}
Intuitively, $\xi(\omega)$ is the ratio of the probability weight of path $\omega$ under $\mathbb{Q}$
relative to $\mathbb{P}$. This is the operational basis of the entire measure-change framework.

The discounted price of a risky asset $\tilde S_t = e^{-rt}S_t$ is generally not a martingale under the
physical measure $\mathbb{P}$ (risk-averse investors demand a positive excess return). The measure
change eliminates this drift by reweighting paths: let $\xi_t = \mathbb{E}^{\mathbb{P}}[\xi_T\mid\mathcal{F}_t]$;
then $\tilde S_t$ being a martingale under $\mathbb{Q}$ is equivalent to $\{\xi_t\tilde S_t\}$ being a
martingale under $\mathbb{P}$. In the continuous-diffusion case, Girsanov's theorem gives
$\xi_t = \exp\!\big(\!-\!\int_0^t\theta_s\,\mathrm{d}W_s^{\mathbb{P}} - \tfrac12\int_0^t\theta_s^2\,\mathrm{d}s\big)$,
where $\theta_t=(\mu_t-r)/\sigma_t$ is the market price of risk. By the fundamental theorem of asset
pricing~\citep{delbaen1994general}, the market is arbitrage-free (NFLVR) if and only if an equivalent
martingale measure exists, in which case the no-arbitrage price of any attainable contingent claim is
$V_0 = \mathbb{E}^{\mathbb{Q}}[e^{-rT}X_T] = \mathbb{E}^{\mathbb{P}}[\xi_T\,e^{-rT}X_T]$. In
incomplete-market models such as jump diffusions and stochastic volatility, the $\xi$ satisfying the
martingale condition is not unique, and \emph{how to identify from data the measure that the market
actually adopts} is precisely the central problem of this paper.

\subsection{Incomplete markets and the non-uniqueness of the ELMM}

Let the filtration $(\F_t)_{t\in[0,T]}$ on the probability space $(\Omega,\F,\Prob)$ satisfy the usual
conditions, let $S=(S_t)$ be the underlying asset price process, and let $r$ be the risk-free rate.

\begin{definition}[Equivalent local martingale measure]
  A measure $\Q$ is called an equivalent local martingale measure of $(S,r)$ if
  $\Q\sim\Prob$ (equivalent) and the discounted price $e^{-rt}S_t$ is a local martingale under $\Q$.
  The set of all ELMMs is denoted $\mathcal{Q}$.
\end{definition}

By the Delbaen--Schachermayer fundamental theorem, $\mathcal{Q}\neq\varnothing$ is equivalent to NFLVR
(no arbitrage), while $|\mathcal{Q}|=1$ is equivalent to market completeness. In incomplete markets
$|\mathcal{Q}|>1$, and each $\Q\in\mathcal{Q}$ corresponds to a different pricing functional.

The Radon--Nikodym theorem guarantees that for any $\Q\sim\Prob$ there exists an $\F_T$-measurable
random variable
\begin{equation}
  \xi = \frac{\mathrm{d}\Q}{\mathrm{d}\Prob} \geq 0,
  \quad \E^{\Prob}[\xi] = 1,
\end{equation}
such that for any $\F_T$-measurable bounded random variable $X$, $\E^{\Q}[X] = \E^{\Prob}[\xi X]$. Thus
Eq.~\eqref{eq:pricing_formula} is equivalent to
\begin{equation}
  C(K,T) = e^{-rT}\,\E^{\Prob}\!\left[\xi\cdot\left(S_T - K\right)^+\right].
  \label{eq:pricing_rnd}
\end{equation}
The ELMM identification problem is therefore recast as: \emph{determine the $\xi$ satisfying the
condition from the option-price constraints}.

\subsection{ELMM selection criteria}
In incomplete markets, $\mathcal{Q}$ is generally not a singleton, and additional economic or
mathematical criteria are needed to select a unique measure from it; two representative ones are as
follows. The first is the \emph{minimal-entropy equivalent martingale measure} (MEMM):
\citet{frittelli2000minimal} proved that, in the dual framework of exponential-utility maximization, the
optimal risk-neutral measure is the ELMM with minimal relative entropy (KL divergence),
\begin{equation}
  \Q^* = \arg\min_{\Q\in\mathcal{Q}}\, H(\Q\|\Prob),
  \qquad H(\Q\|\Prob) = \E^{\Prob}\!\left[\xi\log\xi\right],
\end{equation}
i.e., the measure ``with the least information difference from the physical measure''; in discrete time
and scenario-set implementations it corresponds to a convex optimization problem with good numerical
properties. The second is the \emph{variance-optimal martingale measure} (VOMM): based on the
mean-variance hedging criterion, one selects the ELMM that minimizes the variance of the hedging
error~\citep{schweizer1996approximation}. Unlike the above explicit variational criteria, this paper's
neural-network method does not specify a target functional of the MEMM/VOMM type a priori, but instead
calibrates $\xi$ with the return-matching loss~\eqref{eq:loss_return} and fixes the scale by sample-mean
normalization at the inference stage; nonetheless, network capacity and regularization still introduce an
implicit inductive bias, avoiding extremely dispersed weights while satisfying the martingale and
price-information constraints.

\subsection{The underlying asset models}
\label{sec:models}

The scalar correction of EMS performs poorly on some models. These models share a common feature: the
true Radon--Nikodym derivative $\xi = \mathrm{d}\Q/\mathrm{d}\Prob$ is path-dependent and cannot be
characterized by applying a uniform scalar scaling to all paths, so EMS produces a systematic pricing
bias. The path dependence arises from two mechanisms. The first is \emph{structural path dependence}:
even when $\Prob$ and $\Q$ share the same model parameters and differ only in drift, the underlying
stochastic mechanism itself---an infinitely active pure-jump process or a history-dependent jump
intensity---makes $\xi(\omega)$ a nonlinear functional of the path rather than a scalar; the CGMY
pure-jump model and the Hawkes self-exciting jump-diffusion model belong to this class. The second is an
\emph{explicit risk premium}: the model parameters of $\Prob$ and $\Q$ genuinely differ, reflecting the
market's pricing of jump, volatility, or L\'{e}vy-measure risk, in which case
$\mathrm{d}\Q/\mathrm{d}\Prob$ cannot be approximated by any constant, and methods that ignore this
distributional difference (EMS included) will produce a systematic error; the three settings MertonJRP,
HestonVRP, and CGMY-JRP belong to this class.

The observable consequences of the two mechanisms at the \emph{pricing} level are not the same. In the
first class, $\Prob$ and $\Q$ share the same jump and volatility law and differ only in drift; although
$\xi(\omega)$ is in theory still a path functional, a correct reweighting (such as the weighted Monte
Carlo of Section~\ref{sec:wmc_baseline}) can match prices exactly at the calibration strikes, and its
path dependence does not translate into observable pricing failure, so we treat it as a
\textbf{drift-only} control group; only the second class genuinely changes the \emph{shape} of the $\Prob$
and $\Q$ distributions. Correspondingly, a setting that changes only the vol-of-vol ($\sigma_v$) is not
an equivalent measure change (the diffusion coefficient must remain invariant under an equivalent
change), and we do not adopt it.

\subsubsection{The CGMY pure-jump model}

The CGMY model~\citep{carr2002fine} is a parametric Lévy process whose jump Lévy measure is
\begin{equation}
  \nu(\mathrm{d}x) = C\,\frac{e^{-M|x|}}{|x|^{1+Y}}\,\mathbf{1}_{x<0}\,\mathrm{d}x
    + C\,\frac{e^{-Gx}}{x^{1+Y}}\,\mathbf{1}_{x>0}\,\mathrm{d}x,
\end{equation}
where $C>0$ controls the activity level, $G,M>0$ control the exponential decay rates of the right and
left jumps respectively, and $Y\in(0,2)$ controls the jump activity ($Y=0.5$ corresponds to finite
variation and infinite activity). CGMY contains no diffusion component and is a pure-jump process in the
strict sense; option prices are computed by the Carr--Madan FFT method. Under the baseline setting,
$\Prob$ and $\Q$ share the same Lévy measure and differ only in drift; but the infinitely active
pure-jump structure makes the martingale transform $\xi(\omega)$ depend on the fine realization of
infinitely many small jumps along the path, and it is a path functional rather than a scalar, an
essential difference from diffusion-dominated models. The experimental parameters are $S_0=100$,
$r=0.05$, $\mu=0.10$, $C=0.50$, $G=8.0$, $M=12.0$, $Y=0.50$.

\subsubsection{The Hawkes self-exciting jump-diffusion model}

The Hawkes model~\citep{hawkes1971spectra} characterizes jump clustering via a self-exciting point
process. Under $\Prob$, the asset price and the stochastic jump intensity satisfy
\begin{align}
  \frac{\mathrm{d}S_t}{S_t} &= (\mu - \lambda_t\bar{k})\,\dt
    + \sigma\,\mathrm{d}W_t + (e^{J_t}-1)\,\mathrm{d}N_t, \\
  \lambda_t &= \lambda_0 + \alpha\int_0^t e^{-\beta(t-s)}\,\mathrm{d}N_s,
\end{align}
where $\alpha/\beta<1$ (the branching ratio) guarantees stationarity, and
$J_t\sim\mathcal{N}(\mu_j,\delta_j^2)$. Because $\lambda_t$ is a functional of the complete jump history
before time $t$, $\xi$ is path-dependent even when $\Prob$ and $\Q$ share the same jump parameters, and
EMS, applying a uniform scalar scaling to all paths, cannot capture this history-dependent structure. The
baseline prices under $\Q$ are computed by Monte Carlo simulation of $2$ million paths. The experimental
parameters are $S_0=100$, $r=0.05$, $\mu=0.12$, $\sigma=0.15$, $\lambda_0=1.0$, $\alpha=1.5$, $\beta=3.0$
(branching ratio $0.5$), $\mu_j=-0.05$, $\delta_j=0.12$.

\subsubsection{Explicit risk-premium settings}

The following three settings construct a genuinely path-dependent $\xi = \mathrm{d}\Q/\mathrm{d}\Prob$ by
explicitly specifying different parameters for $\Prob$ and $\Q$. Paths are simulated under $\Prob$, and
the pricing targets are computed by the analytical formulas (Merton / Heston / CGMY) under the
corresponding $\Q$ parameters.

\emph{Merton jump risk premium (MertonJRP).} The market pricing of jump risk makes jumps under $\Q$ more
frequent and larger: $\Prob$ takes $\lambda=0.50$, $\mu_j=-0.10$, $\delta=0.20$; $\Q$ takes
$\lambda^{\Q}=2.00$, $\mu_j^{\Q}=-0.20$, $\delta^{\Q}=0.25$. Common parameters are $S_0=100$, $r=0.10$,
$\sigma=0.20$.

\emph{Heston volatility risk premium (HestonVRP).} The market pricing of volatility risk changes the
mean-reversion speed and long-run level of the variance process under $\Q$: $\Prob$ takes $\kappa=2.0$,
$\theta=0.04$; $\Q$ takes $\kappa^{\Q}=5.0$, $\theta^{\Q}=0.016$, satisfying
$\kappa\theta=\kappa^{\Q}\theta^{\Q}$. Common parameters are $S_0=100$, $r=0.03$, $v_0=0.04$,
$\sigma_v=0.30$, $\rho=-0.70$.

\emph{CGMY Lévy-measure risk premium (CGMY-JRP).} The market pricing of negative-jump risk changes the
Lévy-measure parameters under $\Q$, raising the overall activity and concentrating mass toward the left
tail: $\Prob$ takes $C=0.50$, $G=8.0$, $M=12.0$; $\Q$ takes $C^{\Q}=0.80$, $G^{\Q}=4.0$, $M^{\Q}=18.0$.
$Y=0.50$ is the same under both measures, and the common parameters are $S_0=100$, $r=0.05$.

\subsection{Terminal-scalar sufficiency for European option pricing}
\label{sec:terminal}

In theory, the Radon--Nikodym derivative is a martingale process:
\begin{equation}
  \xi_t = \E^{\Prob}\!\left[\left.\frac{\mathrm{d}\Q}{\mathrm{d}\Prob}
  \right|\mathcal{F}_t\right], \quad t\in[0,T],
\end{equation}
updated continuously at each time node as information accumulates. In the continuous-diffusion case,
Girsanov's theorem gives its explicit stochastic-exponential form:
\begin{equation}
  \xi_t = \exp\!\left(-\int_0^t\theta_s\,\mathrm{d}W_s^{\Prob}
  -\frac{1}{2}\int_0^t\theta_s^2\,\mathrm{d}s\right),
\end{equation}
which is a time path, not a terminal scalar.

For European option pricing, however, only the terminal value $\xi_T$ is needed:
\begin{equation}
  C(K,T) = \E^{\Q}\!\left[e^{-rT}(S_T-K)^+\right]
          = \E^{\Prob}\!\left[\xi_T\cdot e^{-rT}(S_T-K)^+\right].
  \label{eq:pricing_terminal}
\end{equation}
Since the payoff depends only on $S_T$, learning the terminal scalar $\xi_T(\omega)$ is in theory
sufficient to price \emph{the European options of that maturity}. But this word ``sufficient'' already
implies an identification boundary, which the next subsection makes precise.

\subsection{Identifiability: which part of the measure option data can determine}
\label{sec:identification}

Let $\Q\sim\Prob$, $\xi_T=\mathrm{d}\Q/\mathrm{d}\Prob>0$, $\E^{\Prob}[\xi_T]=1$. This subsection
characterizes: given a set of calibration instruments, to what extent $\xi_T$ (equivalently $\Q$) can be
determined. This provides the theoretical skeleton for the empirical picture that ``classical methods are
already optimal on the marginal problem, while neural methods show value only at the measure level.''

\begin{proposition}[European projection: only the terminal marginal is identified]
\label{prop:european_projection}
For any $\F_T$-measurable $\xi_T>0$ and any European payoff $g(S_T)$,
\begin{equation}
  \E^{\Prob}\!\big[\xi_T\,g(S_T)\big]
  = \E^{\Prob}\!\big[\bar\xi_T(S_T)\,g(S_T)\big],
  \qquad \bar\xi_T(S_T):=\E^{\Prob}[\xi_T\mid S_T].
\end{equation}
Hence all European option prices depend only on the terminal marginal density ratio
$\bar\xi_T=\mathrm{d}\Q_{S_T}/\mathrm{d}\Prob_{S_T}$; two samples with the same terminal $S_T$ but
different paths cannot be distinguished by any European option.
\end{proposition}
\begin{proof}
Take the conditional expectation with respect to $\sigma(S_T)$ and use the tower property:
$\E^{\Prob}[\xi_T g(S_T)]=\E^{\Prob}\big[\E^{\Prob}[\xi_T\mid S_T]\,g(S_T)\big]$. \
\end{proof}

\begin{proposition}[Partial identification and the identification wall]
\label{prop:partial_id}
Given call prices at finitely many strikes $\{K_j\}_{j=1}^m$ on maturity $T$, the terminal pricing
measures compatible with $\bar\xi_T$ form a \emph{convex set}; the price of any claim not lying in
$\mathrm{span}\{1,\,S_T,\,(S_T-K_j)^+\}$ is \emph{not determined}. In particular, deep out-of-the-money
($K$ beyond the observed strike range) and \emph{any claim depending on more than the terminal marginal}
(path-dependent / joint-distribution claims) are underidentified. When strikes form a continuum, the
terminal marginal $\Q_{S_T}$ is fully identified by Breeden--Litzenberger, but the path/joint structure
is still not determined.
\end{proposition}
\begin{proof}[Proof sketch]
The price constraints are finitely many linear equalities in $\bar\xi_T$, and their solution set is
convex. The price of a claim $\E^{\Prob}[\bar\xi_T\,\phi]$ is uniquely determined by these constraints if
and only if $\phi$ lies in the linear span of the constraint functionals; otherwise one can vary
$\bar\xi_T$ within the solution set without changing the calibration prices yet changing the price of that
claim. Proposition~\ref{prop:european_projection} has already shown that path information does not enter
$\bar\xi_T$ at all. \
\end{proof}

\begin{proposition}[Risk-specific identification: recovery with matched instruments]
\label{prop:risk_specific}
Let the calibration instruments (claims) $\{h_i\}$ linearly span the subspace $V$. Calibration uniquely
determines the $\Q$-pricing functional $\phi\mapsto\E^{\Q}[e^{-r\cdot}\phi]$ on $V$: claims in $V$ are
priced exactly, and those outside $V$ are not determined. Hence, to identify a certain type of structure,
one must introduce instruments that \emph{span that structure}---second-moment (variance) swaps identify
terminal variance (the tail of the volatility risk premium), third-moment (skewness) swaps identify jump
asymmetry (the tail of the jump risk premium), and path-dependent instruments (such as forward-start)
identify the corresponding joint/path functionals.
\end{proposition}

\begin{proposition}[Measure self-consistency: a single measure vs.\ per-marginal]
\label{prop:consistency}
A $\xi_T$ defined on path space induces a \emph{single} measure $\Q$, which gives no-arbitrage, mutually
consistent prices for \emph{any} claim (including multi-maturity, path-dependent claims). Conversely,
marginal density ratios $\{\bar\xi_{T_j}\}$ calibrated separately for each maturity do not in general come
from the same path-space measure (unless the additional consistency constraint
$\E^{\Prob}[\xi_{T_{j+1}}\mid\F_{T_j}]=\xi_{T_j}$ is imposed), and thus cannot price claims depending on
the joint distribution.
\end{proposition}

\begin{remark}[Correspondence of theory and experiment]
The above propositions provide the skeleton for the empirical picture below, with three correspondences.
Propositions~\ref{prop:european_projection} and~\ref{prop:partial_id} show that European prices identify
only the terminal marginal and that deep out-of-the-money is underidentified. This is precisely the
identification wall observed in Sections~\ref{sec:wmc_baseline}, \ref{sec:xinet_vs_wmc}: since reliable
extrapolation is impossible, minimum-relative-entropy WMC is already the optimal baseline.
Proposition~\ref{prop:risk_specific} gives ``risk-specific identification''---to recover a certain type of
structure, one must introduce instruments that span it, corresponding to the matched moment instruments
recovering the tail in Section~\ref{sec:wmc_baseline}. Proposition~\ref{prop:consistency} is the
theoretical basis of the main positive result (Section~\ref{sec:measure_identification}): per-marginal
calibration does not in general come from the same measure and cannot price path-dependent claims, whereas
a single $\xi$ on path space can. When the true $\xi$ is a smooth functional of the path, the
parameterization $\xi=f_\theta(\text{path features})$ is a \emph{correctly specified} structural prior,
and thus recovers the joint structure that European marginals cannot constrain, outperforming the
reweighting that takes maximum entropy over the same marginal.
\end{remark}

XiNet therefore uses path summary statistics as features to parameterize a path-level positive weight
$\xi(\omega)$, which can both match the marginal when only European constraints are available and identify
the joint structure when path instruments are given.

\section{Method: neural networks learning $\xi=\mathrm{d}\Q/\mathrm{d}\Prob$}
\label{sec:method}

\subsection{Problem formulation}

Given $n$ terminal-price paths $\{S_T^{(i)}\}_{i=1}^n$ simulated under the $\Prob$ measure, with
corresponding path feature vectors $\phi_i\in\R^d$, and benchmark call prices at $m$ training strikes
$\{C(K_j,T)\}_{j=1}^m$ (given by the analytical or numerical formula of the chosen ELMM), our goal is to
learn a function $\xi_\theta:\R^d\to\R_{>0}$ such that the weighted expectation using $\xi_\theta$ as the
Radon--Nikodym derivative is consistent with the $\Q$-measure pricing at the training strikes.

At inference, the network output is normalized:
$\tilde\xi_i = \xi_{\theta,i} \big/ \frac{1}{n}\sum_j \xi_{\theta,j}$, thereby enforcing
$\frac{1}{n}\sum_i \tilde\xi_i = 1$ (the sample mean equals $1$, consistent with the total-mass
constraint $\E^{\Prob}[\xi]=1$), and the option price is estimated with the weighted mean.

\subsection{The return-matching constraint}
\label{sec:return_def}

\paragraph{Definition of return.}
For the underlying asset and a call option with maturity $T$, define the \textbf{normalized discounted
return}:
\begin{align}
  R_S &:= e^{-rT}\,\frac{S_T}{S_0},
  \label{eq:return_stock}\\
  R_C(K_j) &:= e^{-rT}\,\frac{(S_T-K_j)^+}{C(K_j,T)},
  \quad j=1,\ldots,m.
  \label{eq:return_call}
\end{align}
$R_S$ is the normalized total return of holding the discounted stock from $0$ to $T$; $R_C(K_j)$ is the
normalized total return of holding one call option to maturity. Normalization is relative to the current
price ($S_0$ or $C(K_j,T)$), so that the constraint targets across different maturities and strikes are
unified into
\begin{equation}
  \E^{\Q}[R_S] = 1, \qquad \E^{\Q}[R_C(K_j)] = 1.
  \label{eq:return_target}
\end{equation}
The first equality is the martingale condition of the discounted stock under no dividends
$e^{-rT}\E^{\Q}[S_T] = S_0$; the second is the no-arbitrage pricing
$e^{-rT}\E^{\Q}[(S_T-K_j)^+] = C(K_j,T)$.

\paragraph{From $\Q$ to $\Prob$.}
Using $\E^{\Q}[\cdot] = \E^{\Prob}[\xi\cdot]$, Eq.~\eqref{eq:return_target} is equivalent to
\begin{align}
  \E^{\Prob}\!\left[\xi_\theta \cdot R_S\right] &= 1,
  \label{eq:mart_return}\\
  \E^{\Prob}\!\left[\xi_\theta \cdot R_C(K_j)\right] &= 1,
  \quad j=1,\ldots,m.
  \label{eq:call_return}
\end{align}
Eqs.~\eqref{eq:mart_return}--\eqref{eq:call_return} are the training targets of XiNet: using path samples
under the $\Prob$ measure, soft constraints force $\xi_\theta$ to satisfy the above expectation
equalities.

\subsection{Loss function}

Writing Eqs.~\eqref{eq:mart_return}--\eqref{eq:call_return} into the loss as weighted squared errors:
\begin{equation}
  \mathcal{L}(\theta) =
  \lambda_S
  \underbrace{\left(
    \frac{1}{n}\sum_i \xi_{\theta,i}\, R_{S,i} - 1
  \right)^{\!2}}_{\mathcal{L}_S}
  +\;
  \frac{\lambda_C}{m}
  \sum_{j=1}^m
  \underbrace{\left(
    \frac{1}{n}\sum_i \xi_{\theta,i}\, R_{C,i}(K_j) - 1
  \right)^{\!2}}_{\mathcal{L}_{C,j}},
  \label{eq:loss_return}
\end{equation}
where $\xi_\theta$ is the unnormalized network output, with hyperparameters $\lambda_S=5$, $\lambda_C=15$.
Note that $\E^\Prob[\xi_\theta]=1$ is not written into the gradient terms but is realized at inference
through post-hoc normalization (see \S\ref{sec:inference}).

The normalization design (with $1$ as the unified target) avoids the gradient imbalance that arises when
targeting price levels directly, since deep out-of-the-money option prices are extremely small.

\subsection{Network architecture}

$\xi_\theta$ is parameterized by a fully connected feedforward network:
\begin{equation}
  \xi_\theta(\phi) = \mathrm{softplus}\bigl(f_\theta(\phi) + 0.5\bigr),
  \quad \mathrm{softplus}(x) = \log(1+e^x),
\end{equation}
strictly guaranteeing output positivity. The network architecture is $d\to160\to128\to96\to64\to1$
($d=8$-dimensional feature input), each hidden layer containing BatchNorm + ReLU, with parameters
initialized by Xavier initialization (gain $0.5$), so that the initial output concentrates around $1$. The
path features $\phi_i$ are standardized within each training batch (subtract the in-batch mean, divide by
the in-batch standard deviation).

\subsection{Training}

\begin{algorithm}[H]
  \caption{XiNet training}
  \label{alg:training}
  \begin{algorithmic}[1]
    \REQUIRE maturity $T$, training strikes $\{K_j\}_{j=1}^m$,
             hyperparameters $\lambda_S,\lambda_C$,
             number of epochs $N_{\mathrm{epoch}}$, number of paths $n$
    \ENSURE optimal parameters $\theta^*$
    \STATE compute the benchmark price $C_j = C(S_0,K_j,T)$ at each training strike
    \FOR{$e = 1,\ldots,N_{\mathrm{epoch}}$}
      \STATE resample $n$ paths from $\Prob$, extract and standardize features $\Phi\in\R^{n\times d}$
      \STATE forward pass: $\boldsymbol\xi = \xi_\theta(\Phi)\in\R^n_+$
      \STATE compute $\mathcal{L}(\theta)$ by Eq.~\eqref{eq:loss_return}
      \STATE gradient clipping ($\|\nabla\|\leq 1$), AdamW update, cosine-annealing learning rate
      \IF{$\mathcal{L}(\theta) < \mathcal{L}^*$}
        \STATE $\theta^*\leftarrow\theta$, $\mathcal{L}^*\leftarrow\mathcal{L}(\theta)$
      \ENDIF
    \ENDFOR
    \RETURN $\theta^*$
  \end{algorithmic}
\end{algorithm}

Each training epoch resamples paths from $\Prob$ (online learning), preventing overfitting to a specific
path realization and improving generalization to the path distribution. Hyperparameter summary:
$n=22{,}000$ (number of training paths), $N_{\mathrm{epoch}}=3{,}000$, $\eta=1.25\times10^{-3}$ (AdamW
initial learning rate), cosine-annealed to $\eta_{\min}=10^{-6}$, gradient-clipping threshold $1.0$,
trained independently for each ELMM and maturity.

\subsection{Inference: weighted pricing}
\label{sec:inference}

For any option $(K',T)$ to be priced, generate $n_{\mathrm{test}}=55{,}000$ test paths, first normalize
\begin{equation}
  \tilde\xi_i = \frac{\xi_{\theta^*}(\phi_i)}{\dfrac{1}{n_{\mathrm{test}}}\sum_{j=1}^{n_{\mathrm{test}}}\xi_{\theta^*}(\phi_j)},
\end{equation}
then compute the weighted mean:
\begin{equation}
  \hat{C}(K',T) = e^{-rT}\,\frac{1}{n_{\mathrm{test}}}
  \sum_{i=1}^{n_{\mathrm{test}}} \tilde\xi_i\cdot(S_T^{(i)}-K')^+.
  \label{eq:pricing}
\end{equation}
Normalization ensures $\frac{1}{n}\sum_i\tilde\xi_i=1$, making Eq.~\eqref{eq:pricing} converge uniformly
to $e^{-rT}\E^{\Q}[(S_T-K')^+]$ as $n_{\mathrm{test}}\to\infty$.

\subsection{Model-free path features}
\label{sec:mf_features}

When the parametric model is known, one may use internal model variables (number of jumps, diffusion
increments, etc.) as path features. To make the same framework run without assuming an underlying model,
we construct, directly from the observable daily log-return series
$\{r_t\}_{t=1}^{T_{\mathrm{steps}}}$ ($T_{\mathrm{steps}}$ being the number of trading days corresponding
to the annualized maturity $T$, counting 252 days per year), an 8-dimensional model-free feature vector
$\phi_i\in\R^8$ (Table~\ref{tab:mf_features}).

\begin{table}[H]
  \centering
  \caption{The 8-dimensional model-free path features}
  \label{tab:mf_features}
  \small
  \begin{tabular}{clll}
    \toprule
    \# & Feature & Formula & Information captured \\
    \midrule
    1 & Realized volatility & $\sum_t r_t^2$
      & Total volatility level \\
    2 & Max single-day absolute return & $\max_t |r_t|$
      & Extreme events (jump proxy) \\
    3 & Skewness & $\E[(r-\bar r)^3]/\hat\sigma^3$
      & Distributional asymmetry \\
    4 & Excess kurtosis & $\E[(r-\bar r)^4]/\hat\sigma^4 - 3$
      & Tail thickness \\
    5 & Vol-of-Vol & $\sum_t (r_t^2 - r_{t-1}^2)^2$
      & Volatility time-variation \\
    6 & Squared-return autocorrelation & $\mathrm{corr}(r_t^2,\,r_{t-1}^2)$
      & Volatility clustering \\
    7 & Maximum drawdown & $\max_s\!\bigl(\sum_{t\leq s}r_t\bigr) - \min_s\!\bigl(\sum_{t\leq s}r_t\bigr)$
      & Path extremal behavior \\
    8 & Cumulative log-return & $\sum_t r_t = \log(S_T/S_0)$
      & Overall path drift \\
    \bottomrule
  \end{tabular}
\end{table}

The above features rely on no parametric model assumption and can be computed directly from simulated or
real market daily returns. Their role is to provide the network with a statistical basis for
distinguishing the risk characteristics of paths: if the true $\mathrm{d}\Q/\mathrm{d}\Prob$ depends on
the moment features of the path (such as the number of jumps or the degree of volatility clustering), then
the network can learn the corresponding path-weight patterns through these proxy quantities.

\section{Numerical experiments}
\label{sec:experiments}

\subsection{Experimental setup}
\label{sec:exp_setup}

For each model, the test grid consists of $5\times5=25$ $(T,\,S_0/K)$ combinations: maturities
$T\in\{1/12,\,3/12,\,6/12,\,9/12,\,1\}$ (years), moneyness $S_0/K\in\{1.15,\,1.10,\,1.00,\,0.90,\,0.85\}$.
The training strikes are a subset of the test set (the ATM and ITM grid points), and the model does not
see the true prices of deep OTM options during training. A separate XiNet is trained for each maturity,
sampling $n=22{,}000$ simulated paths per epoch and training for 3000 epochs; at inference,
$n_{\mathrm{test}}=55{,}000$ independent paths are used, and for each grid point the estimate is repeated
$n_{\mathrm{reps}}=120$ times and averaged. The error is measured by the RMSE over the 25 grid points,
$\mathrm{RMSE}=\big[\tfrac{1}{25}\sum_{i=1}^{25}(\hat C_i-C_i^{\mathrm{true}})^2\big]^{1/2}$.

This paper compares four pricing methods. BS calibration backs out the Black--Scholes implied volatility
at the training strikes, then interpolates or extrapolates to the full grid. Crude MC takes
$\hat C=e^{-rT}\E^{\Prob}[(S_T-K)^+]$, i.e., unweighted Monte Carlo under $\Prob$. EMS applies a scalar
scaling to the $\Prob$ paths so that $\E^{\Prob}[\tilde S_T]=S_0 e^{rT}$ and then prices. XiNet learns the
path-dependent $\xi=\mathrm{d}\Q/\mathrm{d}\Prob$ via a neural network from the 8-dimensional model-free
path features (Section~\ref{sec:method}). For the risk-premium settings, Section~\ref{sec:bs_greeks}
additionally introduces a BS-calibration-plus-Greeks baseline, which at the training strikes uses the
market-observable $\partial C/\partial K$ to recover the local slope of the implied volatility smile
before extrapolating.

\subsection{Marginal pricing: WMC as the main baseline and the identification wall}
\label{sec:wmc_baseline}

EMS applies only a scalar scaling to the paths and is not a full measure change; as a pricing method it is
too weak to serve as the main baseline. This paper takes weighted Monte Carlo (WMC) as the main baseline:
subject to satisfying the benchmark price constraints, it minimizes the relative entropy to the physical
prior, assigning each path a path-specific probability weight, and its dual solution is the exponential
family
\begin{equation}
  w_i(\lambda) \;=\; \frac{p_i\,\exp(\lambda^\top h_i)}{\sum_\ell p_\ell\,\exp(\lambda^\top h_\ell)},
  \label{eq:wmc_weights}
\end{equation}
where $h_i$ is the value of each benchmark instrument (discounted stock price, call payoff) on path $i$,
and $\lambda$ is obtained by a low-dimensional convex dual that makes the sample constraints hold exactly.
EMS corresponds to the special case where $w_i$ degenerates to a constant, so this paper demotes it to a
diagnostic control.%
\footnote{With the uniform prior $p_i\equiv 1/n$, Eq.~\eqref{eq:wmc_weights} is the classical
minimum-relative-entropy WMC; with the positive weights $a_i=\exp\{f_\theta(\phi_i)\}$ given by the neural
network as the prior, the same projection is the ``network prior + exact projection'' algorithm, see
Section~\ref{sec:method}.}

To distinguish whether the earlier large multiple relative to EMS comes from a methodological gap or from
interpolation, at $T=3$M and over 30 repetitions we fit WMC under two protocols for each setting: under
the FULL protocol the calibration strikes cover all test moneyness (interpolation), and under the HELD-OUT
protocol only ITM and ATM ($K\le S_0$) are used for calibration while deep OTM falls out-of-sample
(extrapolation). Table~\ref{tab:wmc_heldout} reports the relative bias at deep OTM ($S_0/K=0.85$). Under
FULL interpolation the bias of all settings is about zero: as long as the OTM strike falls within the
calibration range, WMC hits it exactly, so the earlier large multiple relative to EMS is mainly an
interpolation gap rather than extrapolation ability. Under HELD-OUT extrapolation an error appears for all
settings, most severe for the shape-risk-premium settings; the root cause is that European call prices
cannot identify the out-of-sample tail, which holds for any method, WMC included. The deep-OTM percentage
is amplified by an extremely small denominator (e.g., the HestonVRP truth is about $0.03$), so the main
text and the appendix report both the absolute error and the error normalized by vega or the bid--ask
spread.

\begin{table}[H]
  \centering
  \caption{Relative pricing bias of the WMC baseline at deep OTM ($S_0/K=0.85$, $T=3$M),
  $(\hat C - C^{\mathrm{true}})/C^{\mathrm{true}}\times100\%$, FULL (interpolation) vs.\ HELD-OUT
  (extrapolation). Shape-RP are the risk-premium settings that genuinely change the shape of the
  $\Prob$/$\Q$ distributions; drift-only is the control group in which $\Prob$ and $\Q$ share the same law
  and differ only in drift. ESS is the average effective sample proportion under HELD-OUT.}
  \label{tab:wmc_heldout}
  \small
  \begin{tabular}{llrrr}
    \toprule
    Category & Setting & FULL\% & HELD-OUT\% & ESS \\
    \midrule
    \multirow{3}{*}{Shape-RP}
      & MertonJRP & $-0.1$ & $+9.4$   & $39\%$ \\
      & HestonVRP & $+0.3$ & $+104.8$ & $92\%$ \\
      & CGMY-JRP  & $+0.1$ & $+169.3$ & $34\%$ \\
    \midrule
    \multirow{2}{*}{Drift-only}
      & CGMY pure jump & $+0.1$ & $-11.2$ & $20\%$ \\
      & Hawkes    & $+0.0$ & $-18.2$ & $93\%$ \\
    \bottomrule
  \end{tabular}
\end{table}

Under the same strict held-out protocol we can further compare XiNet with WMC. The two see exactly the
same European call constraints on the calibration side, with the only difference that XiNet additionally
sees the 8-dimensional path features; if these features carry tail information that the price constraints
could not identify, XiNet should stably outperform WMC. Table~\ref{tab:xinet_vs_wmc} gives the results for
three shape-risk-premium settings at $T=3$M, each trained independently over 6 random seeds.%
\label{sec:xinet_vs_wmc}%
The conclusions are robust to the random seed: EMS has an extremely large bias ($-69\%$ to $+3640\%$) and
can only serve as a diagnostic control; XiNet does not consistently outperform WMC---on MertonJRP WMC
stably wins, with the $\pm1$ standard-deviation bands of the two fully separated, and XiNet instead
underestimates the jump tail by about $28\%$; on HestonVRP XiNet has a better mean but a large variance
(about $\pm40\%$); and on CGMY-JRP the two are close and both markedly deviate from the truth. When all
calibration instruments are European functions of $S_T$, the path features provide only an inductive bias
correlated with the data-generating process, not identification of the out-of-sample tail, consistent
with Section~\ref{sec:terminal} (European pricing depends only on the terminal $S_T$). To stably surpass
WMC, one must introduce identifying instruments that can see the path (such as VIX, variance swaps,
multi-observation-day constraints), which is also the empirical justification for this paper's dynamic
ELMM framework.

\begin{table}[H]
  \centering
  \caption{Relative pricing bias at deep OTM for XiNet ($6$ random seeds), WMC, and EMS under strict
  held-out (calibration only with $K\le S_0$), $(\hat C - C^{\mathrm{true}})/C^{\mathrm{true}}\times100\%$,
  $T=3$M. XiNet reports the mean $\pm$ standard deviation over $6$ seeds; WMC/EMS report the mean $\pm$
  standard deviation over $20$ resamplings. ``Verdict'' compares the $|$mean bias$|$ of XiNet and WMC and
  notes whether their $\pm1$ s.d.\ bands separate.}
  \label{tab:xinet_vs_wmc}
  \small
  \begin{tabular}{llrrrl}
    \toprule
    Setting & $S_0/K$ & XiNet\% & WMC\% & EMS\% & Verdict \\
    \midrule
    \multirow{2}{*}{MertonJRP}
      & 0.90 & $-11.1\pm2.7$  & $+2.7\pm0.7$   & $-58.5$   & WMC (bands separate)\\
      & 0.85 & $-27.7\pm6.2$  & $+9.4\pm2.4$   & $-68.6$   & WMC (bands separate)\\
    \midrule
    \multirow{2}{*}{HestonVRP}
      & 0.90 & $-2.4\pm10.6$  & $+30.0\pm1.2$  & $+82.7$   & XiNet \\
      & 0.85 & $+15.4\pm40.0$ & $+106.2\pm4.7$ & $+227.7$  & XiNet (very high variance)\\
    \midrule
    \multirow{2}{*}{CGMY-JRP}
      & 0.90 & $+50.8\pm2.5$  & $+63.3\pm1.1$  & $+1210.7$ & XiNet (marginal)\\
      & 0.85 & $+154.1\pm6.8$ & $+169.9\pm3.3$ & $+3639.7$ & XiNet (marginal, mean diff)\\
    \bottomrule
  \end{tabular}
\end{table}

\subsection{From marginals to the measure: path-dependent claims}
\label{sec:measure_identification}

The conclusion of the previous section covers only the marginal problem: a European option is a function
of $S_T$, and given the marginal constraints minimum-relative-entropy WMC is already the optimal
reweighting, which a neural network cannot surpass. But the essence of identifying the market pricing
measure is the ability to price any claim, including exotic claims depending on the path or joint
distribution, and this is exactly where the value of learning a per-path $\xi=\mathrm{d}\Q/\mathrm{d}\Prob$
shows. Consider an at-the-money forward-start call with maturity $T_2=6$M striking at $T_1=3$M, with
payoff $\big(S_{T_2}/S_{T_1}-1\big)^+$, which depends on the joint distribution of $(S_{T_1},S_{T_2})$
rather than on any single marginal. The three methods calibrate on the exact same European option surface
(near-ATM strikes at $T_1$ and $T_2$): XiNet uses a single $\xi=f_\theta(\text{full-path features})$ to
give a self-consistent measure; WMC-joint fits the minimum-relative-entropy measure in one shot on the
same path pool using the European constraints of $T_1$ and $T_2$, equally self-consistent; WMC-per-maturity
fits the two marginals separately, without a joint distribution, and can only give a proxy price under an
independence assumption. Table~\ref{tab:forward_start} reports the bias relative to the truth
($\Q$-Monte-Carlo, cross-validated against the closed form under stationary increments).

\begin{table}[H]
  \centering
  \caption{Relative pricing bias of the at-the-money forward-start option $(S_{T_2}/S_{T_1}-1)^+$
  ($T_1=3$M, $T_2=6$M), $(\hat P-P^{\mathrm{true}})/P^{\mathrm{true}}\times100\%$. All three methods
  calibrate on the same near-ATM European surface at $T_1,T_2$. XiNet reports the mean $\pm$ standard
  deviation over $5$ random seeds; WMC reports $15$ resamplings. ``E'' is European-only calibration, and
  ``E+FS'' is calibration with one additional forward-start constraint.}
  \label{tab:forward_start}
  \small
  \begin{tabular}{llrrr}
    \toprule
    Setting & Calibration & XiNet\% & WMC-joint\% & WMC-per-maturity\% \\
    \midrule
    \multirow{2}{*}{MertonJRP}
      & E    & $+0.3\pm1.4$ & $-24.3\pm0.5$ & $+117.2$ \\
      & E+FS & $-0.9\pm1.5$ & $+0.0$        & --- \\
    \midrule
    \multirow{2}{*}{HestonVRP}
      & E    & $+23.2\pm1.4$ & $+31.1\pm0.3$ & $+104.6$ \\
      & E+FS & $-0.9\pm0.5$  & $-0.0$        & --- \\
    \bottomrule
  \end{tabular}
\end{table}

Per-maturity WMC has a bias of $+105\%$ to $+117\%$: marginal measures cannot price a path claim, and only
a single joint measure can. Under European-only calibration the bias of XiNet is smaller than that of
WMC-joint (MertonJRP: $+0.3\%$ vs.\ $-24.3\%$; HestonVRP: $+23.2\%$ vs.\ $+31.1\%$, with the $\pm1$
standard-deviation bands of the two not overlapping). The two calibrate on the same European surface and
both fit the marginal, so the difference comes from the joint structure: maximum entropy takes the least
informative joint given the marginal, which deviates when the true $\xi$ has feature-driven path
dependence, whereas $\xi=f_\theta(\text{path features})$ is a correctly specified structural prior; the
magnitude of the advantage is risk-specific, larger on jump-driven path dependence and smaller on
volatility-driven ones. European-only is still underidentified for the path claim (HestonVRP is still
biased $+23\%$), and once a matched path instrument (one forward-start constraint) is added, both methods
recover to within $\pm1\%$. This is consistent with the risk-specific identification of the previous
section: the terminal marginal is identified by European options, the tail by moment or variance
instruments, and the joint and path by path-dependent instruments, while XiNet is the single measure that
absorbs any instrument and prices all claims consistently.%

\subsection{Failure mechanisms, cross-model robustness, and diagnostics}
\label{sec:mechanisms}
\label{sec:risk_premium}

This section returns to the mechanism level, first dissecting the failure source of each baseline with a
single risk-premium model, then testing whether Greek information can compensate for the extrapolation
deficiency, and finally giving the two-regime pattern with the breadth of 10 models along with a
diagnostic for judging the strength of path dependence.

Take MertonJRP as an example ($\Prob$: $\lambda=0.5,\ \mu_j=-0.10,\ \delta=0.20$; $\Q$:
$\lambda=2.0,\ \mu_j=-0.20,\ \delta=0.25$). Here $\mathrm{d}\Q/\mathrm{d}\Prob$ is no longer a scalar, and
the failure mechanism of each baseline differs. The error of BS calibration comes from extrapolation: the
training strikes cover only ATM and ITM, and the backed-out implied volatility reflects only the $\Q$
distribution near these grid points, while the $\Q$ distribution of MertonJRP contains high-frequency
negative jumps, and the true short-maturity smile rises again in the deep-OTM direction due to jump
convexity (Figure~\ref{fig:iv_smile}); any linear extrapolation extends along the downward direction of
the training range and cannot capture the curvature reversal, underestimating by more than $60\%$ at
$T=1$M, $S_0/K=0.85$. The error of EMS and Crude MC comes from the distributional shape: the jump
frequency under $\Prob$ is only a quarter of that under $\Q$, the EMS scalar coefficient is close to $1$
and barely changes the distributional shape, while OTM pricing depends on the far tail of the
distribution and the far-tail mass of $\Prob$ is insufficient, leading to a deep-OTM underestimation of
about $77\%$, comparable to Crude MC's $75\%$. XiNet, on the other hand, distinguishes jump-dense from
jump-sparse paths from features such as realized volatility, skewness, and maximum absolute return, learns
that jump-dense paths should receive higher weight, and keeps the OTM bias generally within $\pm1\%$ (only
about $-8\%$ at deep OTM for $T=1$M). Figure~\ref{fig:mertonjrp_error} and
Table~\ref{tab:jump_rp_detail} give a systematic comparison of the methods under this setting.

\begin{figure}[t]
  \centering
  \includegraphics[width=\textwidth]{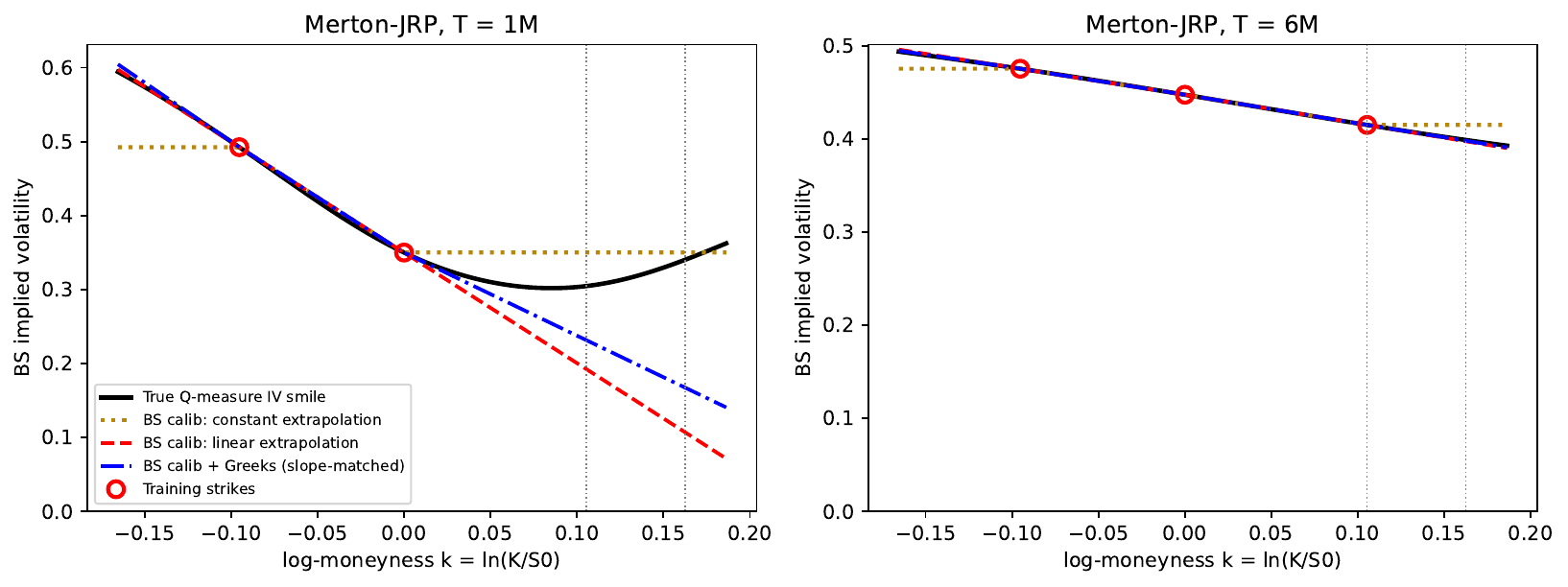}
  \caption{Extrapolation failure of the BS implied volatility smile under the Merton-JRP model
    (left: $T=1$M; right: $T=6$M). The solid black line is the full IV smile backed out from the true
    $\Q$-measure pricing formula; the red circles are the training grid points, the red dashed line is the
    linear extrapolation of the training-grid IV, and the blue dash-dotted line is the slope-matched
    extrapolation after additionally using the training-grid Greek ($\partial C/\partial K$) to recover
    the local slope; the gray vertical lines mark the deep-OTM test grid points. At $T=1$M the true smile
    rises again in the OTM direction due to jump convexity, both extrapolations extend in the opposite
    direction, and the Greek information can only correct the slope, not recover the out-of-range
    curvature; at $T=6$M the smile is approximately linear and the extrapolation error is correspondingly
    smaller.}
  \label{fig:iv_smile}
\end{figure}

\begin{figure}[t]
  \centering
  \includegraphics[width=0.82\textwidth]{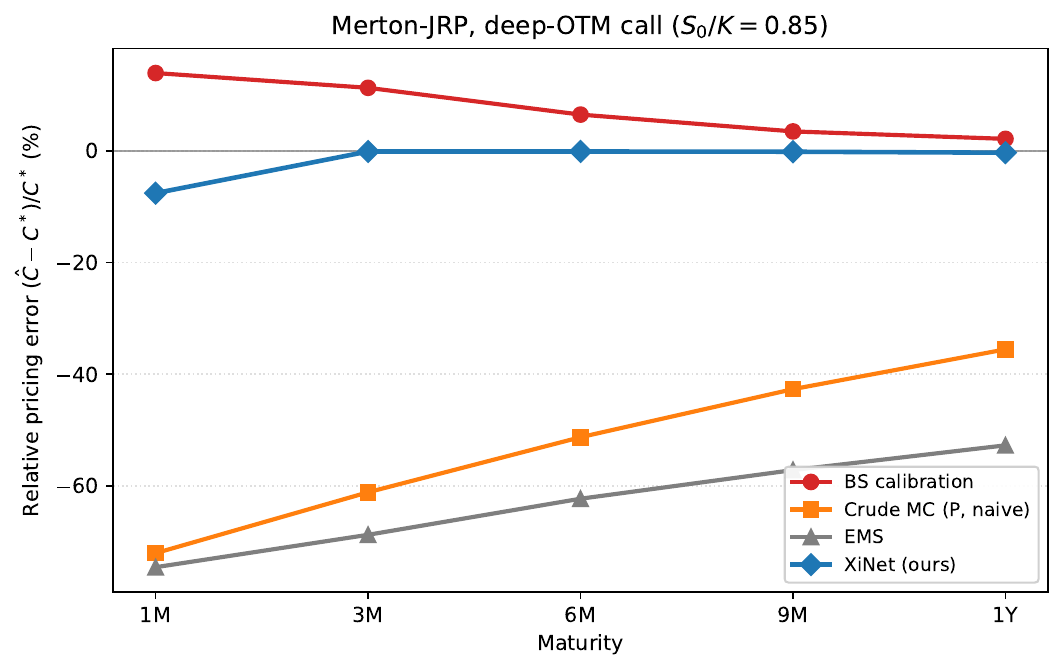}
  \caption{Comparison of pricing errors for an out-of-range call option ($S_0/K = 0.85$) under the
    Merton-JRP model. The horizontal axis is the five maturities (1M--1Y), and the vertical axis is the
    relative pricing bias $(\hat{C} - C^*)/C^* \times 100\%$. BS calibration (red) deviates severely at
    short maturities and converges as the maturity lengthens; Crude MC (orange) and EMS (gray) persistently
    underestimate by about $75\%$--$77\%$, as both fail to capture the $P\to Q$ jump risk premium; XiNet
    (blue) keeps the error within $\pm 2\%$ at all maturities.}
  \label{fig:mertonjrp_error}
\end{figure}

\begin{table}[H]
  \centering
  \caption{MertonJRP setting: comparison of pricing bias for BS calibration (constant extrapolation),
  BS+Greeks, Crude MC, EMS, XiNet.\\
  P: $\lambda=0.5$, $\mu_j=-0.10$, $\delta=0.20$; Q: $\lambda=2.0$, $\mu_j=-0.20$, $\delta=0.25$.
  BS calibration backs out implied volatility at the training strikes $\{S_0/K=1.1,\,1.0\}$ (1M) or
  $\{1.1,\,1.0,\,0.9\}$ (other maturities) and then extrapolates; Crude MC is unweighted simulation under
  $\Prob$. The truth is taken from the corrected Merton formula, and XiNet is trained with the ratio-form
  loss. Bias: $(\hat C - C^{\mathrm{true}})/C^{\mathrm{true}}\times 100\%$.}
  \label{tab:jump_rp_detail}
  \small
  \begin{tabular}{llrrrrrr}
    \toprule
    $T$ & $S_0/K$ & $C^{\mathrm{true}}$ &
      BS const\% & BS+Greeks\% & Crude MC\% & EMS\% & XiNet\% \\
    \midrule
    \multirow{3}{*}{1M}
      & 1.15 & 15.22 & $-2.9$ & $+0.2$ & $-5.5$ & $-8.3$ & $+0.1$ \\
      & 0.90 &  0.63 & $+51.3$ & $-63.5$ & $-59.6$ & $-64.2$ & $-0.5$ \\
      & 0.85 &  0.26 & $+13.9$ & $-99.6$ & $-72.0$ & $-74.5$ & $-7.6$ \\
    \midrule
    \multirow{2}{*}{6M}
      & 1.15 & 22.97 & $-1.1$ & $+0.0$ & $-9.0$ & $-18.9$ & $-0.2$ \\
      & 0.85 &  6.86 & $+6.5$ & $-0.3$ & $-51.3$ & $-62.3$ & $-0.1$ \\
    \midrule
    \multirow{2}{*}{1Y}
      & 1.15 & 29.51 & $-0.6$ & $+0.0$ & $-5.0$ & $-20.4$ & $-0.2$ \\
      & 0.85 & 14.80 & $+2.1$ & $+0.0$ & $-35.5$ & $-52.7$ & $-0.4$ \\
    \bottomrule
  \end{tabular}
\end{table}

The failure of BS calibration comes from the curvature outside the training range being unconstrainable by
information within the range (even when supplemented with the slope provided by Greeks), and the failure of
EMS and Crude MC comes from the tail mass missing from the $\Prob$ distribution being irreproducible by a
scalar correction; XiNet makes no prior assumption about the smile or the distributional shape, and both
types of bias are avoided.%
\label{sec:bs_greeks}%
To test whether BS calibration merely uses too little market information, we supplement the training
strikes with Greeks ($\partial C/\partial K$, which can be estimated in the market from adjacent quotes),
recovering the local slope of the smile at the training boundary via
$\partial C/\partial K=\partial C_{\mathrm{BS}}/\partial K|_\sigma
+\mathrm{vega}\cdot\partial\sigma_{\mathrm{iv}}/\partial K$ and then extrapolating.
Table~\ref{tab:bs_rmse_summary} shows that slope matching improves the full-grid RMSE by an average of
$25\%$ to $50\%$, but still fails at short-maturity deep OTM (MertonJRP: bias $-99.6\%$ at $T=1$M,
$S_0/K=0.85$): first-order information cannot recover the curvature outside the training range, and the
pricing information of a short-maturity jump model is exactly concentrated in the curvature. On the
diffusion-dominated HestonVRP with an approximately linear smile, the interpolation-type BS baseline is
itself of very small error (absolute RMSE $0.002$ to $0.024$), showing that XiNet's comparative advantage
is concentrated in the jump-dominated, short-maturity, deep-OTM extrapolation scenario.

\begin{table}[H]
  \centering
  \caption{RMSE (price units) of BS calibration (naive linear extrapolation) and BS+Greeks (slope-matched
    extrapolation). Parentheses give the relative bias at $T=1$M, $S_0/K=0.85$. The XiNet reference values
    are taken from Table~\ref{tab:baseline_mf}.}
  \label{tab:bs_rmse_summary}
  \small
  \begin{tabular}{lcccc}
    \toprule
    \multirow{2}{*}{Model} &
    \multicolumn{2}{c}{Full test grid RMSE} &
    \multicolumn{2}{c}{OTM-specific ($S_0/K\leq0.90$)} \\
    \cmidrule(lr){2-3}\cmidrule(lr){4-5}
    & BS naive & BS+Greeks & BS naive & BS+Greeks \\
    \midrule
    MertonJRP  & 0.127 ($-100\%$) & 0.098 ($-99.6\%$) & 0.200 & 0.154 \\
    HestonVRP  & 0.0034 ($-20\%$) & 0.0018 ($-37\%$)  & 0.0049 & 0.0028 \\
    \bottomrule
  \end{tabular}
\end{table}

Table~\ref{tab:baseline_mf} and Figure~\ref{fig:rmse_bars} summarize the RMSE of 10 models under XiNet,
EMS, and Crude MC, with EMS still serving as a diagnostic control, and the multiple relative to EMS in the
table reflects the degree of failure of the scalar assumption rather than the magnitude of XiNet's
advantage over WMC. The result is two-regime: on the five $\Prob\approx\Q$ models
$\mathrm{d}\Q/\mathrm{d}\Prob$ is approximately constant, the uniform scalar scaling of EMS is already
optimal, and XiNet is close to Crude MC but not better than EMS; on the five models with a risk premium the
scalar assumption fails, and the RMSE of EMS rises sharply. XiNet does not degrade when $\Prob\approx\Q$
and does not fail under a risk premium. The indicator distinguishing the two situations is
$\xi_{\mathrm{std}}$, the sample standard deviation of XiNet's output $\xi$ (Table~\ref{tab:xi_std}): it is
close to zero when $\mathrm{d}\Q/\mathrm{d}\Prob$ is a scalar and deviates significantly from zero when it
is path-dependent; when $\xi_{\mathrm{std}}\gtrsim0.5$ the path-dependent weights are significant, and the
scalar correction may severely underestimate at deep OTM.%
\label{sec:xi_diag}

\begin{table}[H]
  \centering
  \caption{Summary of the 10-model model-free experiment RMSE (the first three numeric columns are RMSE,
           in price units). Crude MC directly simulates $N{=}200{,}000$ paths under the $Q$ measure with no
           variance reduction, and can be viewed as the benchmark upper bound when the $Q$ measure is known.
           $^\dagger$ CGMY-class models are infinitely active Lévy processes, and step-by-step CDF path
           sampling has significant discretization error; their Crude MC is instead obtained by exact
           sampling from the terminal distribution; under 120 repetitions the single-run RMSE is
           $0.018\pm0.006$ (CGMY) and $0.024\pm0.009$ (CGMY-JRP).}
  \label{tab:baseline_mf}
  \small
  \setlength{\tabcolsep}{4pt}
  \begin{tabular}{l ccccc p{2.8cm}}
    \toprule
    Model & XiNet & EMS & EMS/XiNet & Crude MC & MC/XiNet & Note \\
    \midrule
    \multicolumn{7}{l}{\textit{$P \approx Q$ models (drift-only difference)}} \\
    \midrule
    Merton-Jump    & 0.0362 & 0.0019 & $0.05\times$  & 0.029 & $0.79\times$ & scalar approx.\ holds \\
    Heston         & 0.0255 & 0.0022 & $0.09\times$  & 0.022 & $0.86\times$ & scalar approx.\ holds \\
    VG             & 0.0354 & 0.0047 & $0.13\times$  & 0.033 & $0.93\times$ & scalar approx.\ holds \\
    Beta-Jump      & 0.0261 & 0.0057 & $0.22\times$  & 0.026 & $1.00\times$ & scalar approx.\ holds \\
    Bates          & 0.0172 & 0.0036 & $0.21\times$  & 0.031 & $1.80\times$ & scalar approx.\ holds \\
    \midrule
    \multicolumn{7}{l}{\textit{$P \neq Q$ models (with risk premium)}} \\
    \midrule
    CGMY$^\dagger$      & 0.604  & 19.79  & $32.8\times$  & 0.018 & $0.03\times$ & scalar assumption fails \\
    Heston-VRP          & 0.033  & 0.995  & $30.1\times$  & 0.010 & $0.30\times$ & scalar assumption fails \\
    Merton-JRP          & 0.039  & 4.954  & $127.6\times$ & 0.021 & $0.55\times$ & scalar assumption fails \\
    CGMY-JRP$^\dagger$  & 0.701  & 17.26  & $24.6\times$  & 0.024 & $0.03\times$ & scalar assumption fails \\
    Hawkes              & 0.196  & 0.650  & $3.3\times$   & 0.034 & $0.17\times$ & scalar assumption fails \\
    \bottomrule
  \end{tabular}
\end{table}

\begin{figure}[t]
  \centering
  \includegraphics[width=\textwidth]{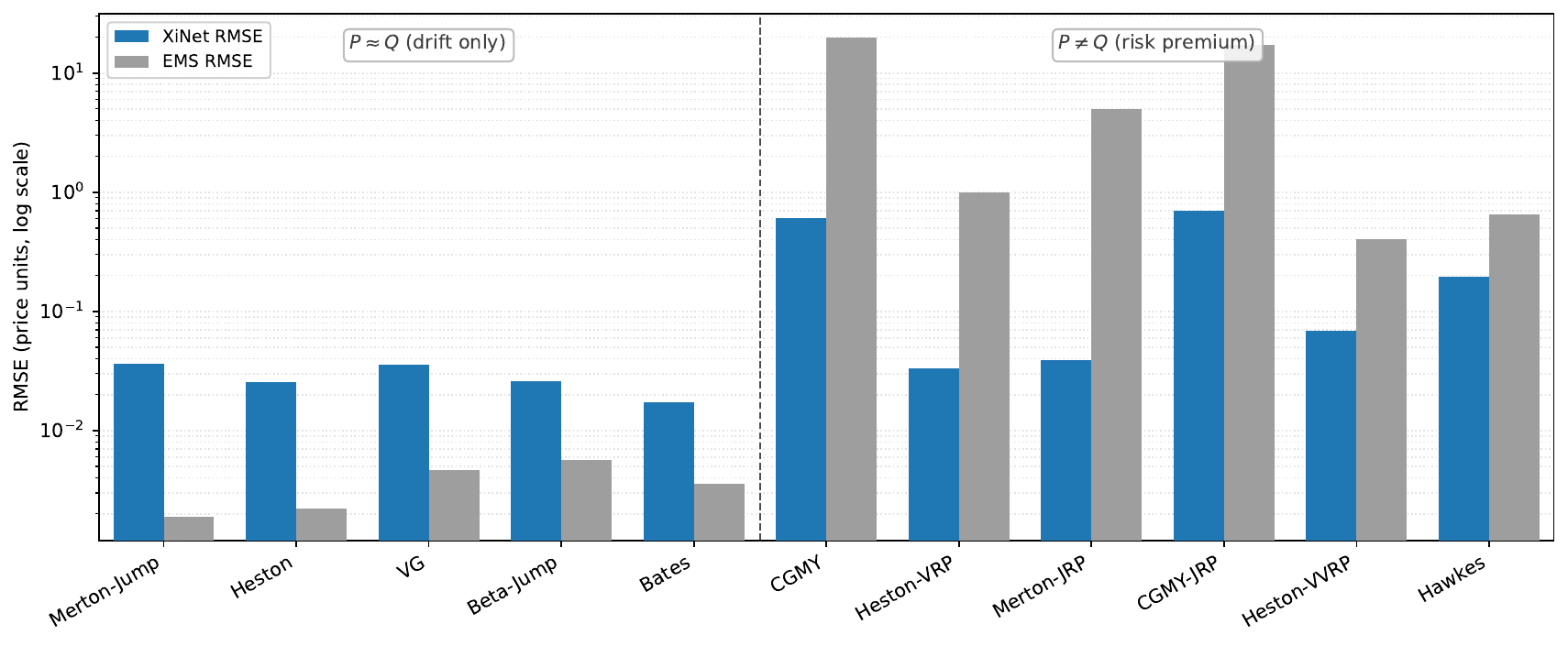}
  \caption{RMSE comparison of 11 stochastic-volatility/jump models under the four pricing methods
    (logarithmic vertical axis). The left 5 models are $P \approx Q$ (drift-only difference): EMS has the
    lowest RMSE, and XiNet is close to Crude MC; the right 6 models contain an explicit risk premium
    ($P \neq Q$): the RMSE of both BS calibration and EMS rises sharply, while XiNet stays at the lowest
    level.}
  \label{fig:rmse_bars}
\end{figure}

\begin{table}[H]
  \centering
  \caption{$\xi_{\mathrm{std}}$ (the sample standard deviation of XiNet's output) as a diagnostic
  indicator of path-dependence strength.}
  \label{tab:xi_std}
  \small
  \begin{tabular}{lccc}
    \toprule
    Model & $\xi_{\mathrm{std}}$ (typical) & EMS/XiNet & Note \\
    \midrule
    Heston-mf (baseline) & $0.25$--$0.42$ & $0.09\times$  & scalar approx.\ holds \\
    Bates-mf (baseline) & $\approx 0.40$ & $0.21\times$  & scalar approx.\ holds \\
    HestonVRP-mf & $0.60$--$0.75$ & $30.1\times$  & path dependence significant \\
    MertonJRP-mf & $0.65$--$1.41$ & $127.6\times$ & path dependence significant \\
    CGMY-mf & $1.3$--$2.7$ & $32.8\times$  & path dependence significant \\
    \bottomrule
  \end{tabular}
\end{table}

\section{Application to real data}
\label{sec:realdata}

The synthetic experiments have an oracle and have already given evidence on the mechanism; real option
data lack a truth for exotic claims, but both halves of the story can be tested on the real SPX. The
marginal half: identification theory predicts that, for European (marginal) pricing,
minimum-relative-entropy WMC is the optimal baseline and XiNet is close to it. The measure half: on
path-dependent claims, a method that only fits marginals should diverge systematically from methods that
identify the path measure (XiNet, WMC-joint). Both learn the $\Prob\to\Q$ measure change from market
prices; the scalar EMS is a finite-sample martingale correction of the Monte Carlo under an already-chosen
$\Q$ model, and there is no simulable $\Q$ on real data, so it is not applicable and is not compared in
this section.

The data are taken from the SPX options of WRDS OptionMetrics. The European surface covers $417$ trading
days from $2024$-$01$ to $2025$-$08$ ($6$ maturity buckets $1$M--$12$M, moneyness $0.85$--$1.15$), split
into train/test by time order, where the test segment is the $165$ trading days of $2025$, used for
day-by-day rolling out-of-sample evaluation; the physical measure $\Prob$ is generated by a block
bootstrap (block length $10$ days) of the $2010$--$2023$ historical daily returns, a semi-parametric
physical prior that does not overlap in time with the market option surfaces. For each (trading day,
maturity bucket), calibration is done with the market mid at near-ATM (moneyness $\{0.90,\dots,1.10\}$) as
constraints, and out-of-sample pricing is done for the two held-out wings ($0.85$, $1.15$).

\subsection{The European (marginal) surface: feasibility and the identification boundary}
\label{sec:realdata_european}

Table~\ref{tab:spx} summarizes the out-of-sample results over all $164$ rolling test days that can be
included in the statistics $\times\{1\text{M},3\text{M},6\text{M}\}$ (a total of $2005$ grid points). To
ensure WMC still gives a stable solution on the near-infeasible surface, a penalized (increasing ridge)
solver is used, and the exact-feasibility rate is counted separately.

\begin{table}[H]
  \centering
  \caption{Real SPX European pricing: average pricing error relative to the market mid
  $|\hat C-\mathrm{mid}|/\mathrm{mid}$ ($164$ rolling test days
  $\times\{1\text{M},3\text{M},6\text{M}\}$, $2005$ grid points). The calibration points are in-sample,
  and the two wings are out-of-sample.}
  \label{tab:spx}
  \small
  \begin{tabular}{lrr}
    \toprule
    & XiNet & WMC \\
    \midrule
    Calibration points (in-sample)      & $2.7\%$  & $0.6\%$ (feasible cell $0.0\%$)\\
    Held-out wings ($0.85$,\,$1.15$)  & $29.8\%$ & $27.2\%$ \\
    \quad of which deep-OTM $0.85$        & $49.5\%$ & $45.7\%$ \\
    \quad of which OTM $1.15$             & $1.7\%$  & $0.8\%$ \\
    \bottomrule
  \end{tabular}
\end{table}

There are four points in the table. On the held-out wings WMC is slightly better and XiNet is close to it
($29.8\%$ vs.\ $27.2\%$; day by day XiNet is better on only $45/139$ days, and at ITM $1.15$ both methods
are within about $2\%$), consistent with the prediction of identification theory for the marginal problem
(Sections~\ref{sec:identification}, \ref{sec:wmc_baseline}): European is a marginal problem,
minimum-relative-entropy WMC is already optimal, and the neural method brings no additional identifying
power, only needing to be close to it. The support and feasibility limitations appear: the proportion of
the market prices that are exactly feasible for the historical bootstrap physical pool is about $84.5\%$,
i.e., for about one sixth of the surface the $2024$--$25$ market risk-neutral tail cannot be obtained by
reweighting the $2010$--$23$ historical pool, in which case exact WMC is infeasible and penalized WMC is
used instead. Deep out-of-the-money $0.85$ is hard for both methods (about $47\%$), because the tail mass
of the semi-parametric physical prior is insufficient, and any method that merely reweights existing
scenarios is bound by the support; its mitigation requires introducing tail or path instruments, or a
better-fitting physical prior (such as volatility-filtered historical simulation). The within-bid--ask hit
rate is comparable for the two methods: SPX quote spreads are extremely narrow (the median relative
half-spread is $0.73\%$ at the calibration end and about $4.5\%$ on the wings), and the proportion of the
two methods' predictions falling within the nearest actual quote bid/ask is low and close (on the held-out
wings XiNet $10.5\%$, WMC $8.8\%$; at the calibration end $15.9\%$ vs.\ $17.7\%$).

\subsection{Path-dependent claims: the forward-start divergence on the real SPX}
\label{sec:realdata_forward_start}

Feasibility at the marginal level is only half the story. The claim of identification theory is that an
identified measure must be able to price path-dependent claims, whereas a method that only fits marginals
cannot. We move the forward-start experiment of Section~\ref{sec:measure_identification} from synthetic
data to the real SPX: the same ATM forward-start $(S_{T_2}/S_{T_1}-1)^+$ ($T_1$ taken as the actual
maturity of the $3$M bucket, $T_2$ that of the $6$M bucket), with the three methods calibrated on the exact
same $3$M and $6$M near-ATM European surface---XiNet and WMC-joint (path self-consistent), and
WMC-per-maturity (fitting the $3$M and $6$M marginals separately and pairing them under an independence
assumption, without a joint distribution).

The real market has no listed truth for the forward-start, so the bias relative to the truth cannot be
reported; but the essence of identification can be tested by the divergence between the methods: if
marginals cannot identify a path claim, then WMC-per-maturity must diverge systematically from the two
path-self-consistent measures. Table~\ref{tab:spx_forward_start} reports, over the $94$ test days on which
both $3$M and $6$M are available, the distribution of the pairwise relative divergence
$(\hat P_A-\hat P_B)/\hat P_B$ between the three methods.

\begin{table}[H]
  \centering
  \caption{Pairwise relative divergence $(\hat P_A-\hat P_B)/\hat P_B\times100\%$ of the ATM forward-start
  $(S_{T_2}/S_{T_1}-1)^+$ ($T_1=3$M, $T_2=6$M) on the real SPX, median and interquartile range over $94$
  test days. The three methods calibrate on the exact same $3$M, $6$M near-ATM European surface; real data
  have no truth, and the divergence itself is the identification signal.}
  \label{tab:spx_forward_start}
  \small
  \begin{tabular}{lrr}
    \toprule
    Relative divergence ($A$ vs $B$) & Median & Interquartile range \\
    \midrule
    XiNet \ vs \ WMC-per-maturity (independent)      & $-46.5\%$ & $[-47.5,\,-45.6]$ \\
    WMC-joint \ vs \ WMC-per-maturity (independent)  & $-42.8\%$ & $[-50.1,\,-39.0]$ \\
    XiNet \ vs \ WMC-joint              & $-7.1\%$  & $[-12.6,\,+7.4]$ \\
    \bottomrule
  \end{tabular}
\end{table}

There are three points in the result. The independence assumption systematically misprices by about
$46\%$: the median divergence of WMC-per-maturity relative to XiNet is $-46.5\%$, negative on all $94$ days
(the $10$/$90$ percentiles are $[-48.6\%,-44.4\%]$), i.e., the independence proxy price that ignores the
joint structure over $[T_1,T_2]$ systematically overestimates this path claim by about $46\%$. The two
path-self-consistent measures corroborate each other: the median divergence between XiNet and WMC-joint is
only $-7.1\%$, far smaller than the divergence of the two from WMC-per-maturity ($-46.5\%$, $-42.8\%$),
showing that the $46\%$ gap comes from the independence assumption itself rather than from the artificial
setting of one method. The divergence is not caused by extreme weights: the effective sample size ESS of
WMC-joint is about $74\%$, the exact-feasibility rate is $99\%$, and the maximum normalized weight is about
$9\times10^{-4}$.

This result extends the claim of Section~\ref{sec:measure_identification} from synthetic experiments to
the real SPX: under exactly the same European calibration, a method that only fits marginals differs from
a method that identifies the path measure by about $46\%$ on a path-dependent claim, while the two
path-self-consistent constructions are close to each other. Although real data have no truth for the path
claim and only a divergence measure can be given, its direction and magnitude are consistent with the
synthetic experiments with an oracle, together supporting this paper's core proposition: what option data
identify on path space is a measure, not merely a few marginals.

\section{Conclusion and outlook}
\label{sec:conclusion}

\subsection{Main conclusions}

This paper proposes XiNet, which learns the Radon-Nikodym derivative $\xi=\mathrm{d}\Q/\mathrm{d}\Prob$
directly via a neural network to identify the market equivalent martingale measure in incomplete markets.
With minimum-relative-entropy WMC as the main baseline, and under the conventions of strict held-out and
multi-seed confidence intervals, this paper's conclusions are as follows.

The core is identifying the pricing measure rather than only the marginal. European pricing is a marginal
problem, on which given the marginal constraints minimum-relative-entropy WMC is already the optimal
reweighting, and XiNet matches rather than surpasses it. But the essence of identifying the market pricing
measure is the ability to price any claim: calibrated on the same European surface, the single
self-consistent measure learned by XiNet has, in the synthetic experiments, a bias of
$+0.3\%\pm1.4\%$ (jump MertonJRP) and $+23.2\%\pm1.4\%$ (volatility HestonVRP) for a path-dependent claim
(the at-the-money forward-start option), smaller than the equally self-consistent maximum-entropy WMC-joint
($-24.3\%$, $+31.1\%$), while a method that only fits marginals fails structurally ($+105\%$ to $+117\%$);
the difference comes from the joint structure that European options cannot constrain, and
$\xi=f_\theta(\text{path features})$ is a correctly specified structural prior
(Section~\ref{sec:measure_identification}). This conclusion is consistently corroborated on the real SPX:
for the same forward-start claim, the independence proxy price that ignores the joint structure diverges
systematically from the two path-self-consistent measures by about $46\%$ (all $94$ trading days in the
same direction), while XiNet and WMC-joint differ from each other by only about $7\%$
(Section~\ref{sec:realdata_forward_start}).

Identification is risk-specific: the terminal marginal is identified by European options, the tail by
moment or variance instruments, and the joint and path by path-dependent instruments, while XiNet is the
single measure that absorbs any instrument and prices all claims consistently. Correspondingly, pure
European is underidentified for the out-of-sample tail, and under strict held-out XiNet does not
consistently outperform WMC (WMC stably wins on MertonJRP): when all calibration instruments are European
functions of $S_T$, the path features provide only an inductive bias, not identification
(Section~\ref{sec:xinet_vs_wmc}). On the real SPX European surface the two methods are comparable on the
held-out wings ($29.8\%$ vs.\ $27.2\%$) and expose the support limitation (the historical pool's
exact-feasibility rate is about $84.5\%$, and deep out-of-the-money is about $47\%$ for both methods)
(Section~\ref{sec:realdata_european}).

EMS is the special case where $\xi$ degenerates to a scalar, and this paper treats it as a diagnostic
rather than the main baseline. When $\Prob$ and $\Q$ differ only in drift, the scalar correction is
already near-optimal and XiNet does not beat it; when an explicit risk premium exists, EMS is comparable to
Crude MC and systematically underestimates deep OTM. The strength of path dependence can be diagnosed by
$\xi_{\mathrm{std}}$ (the sample standard deviation of XiNet's output): when $\xi_{\mathrm{std}}\gtrsim0.5$
the path weights are significant and the scalar correction may severely underestimate. The 8-dimensional
model-free features (realized volatility, maximum absolute return, skewness, kurtosis, etc.) can capture
the jump and volatility heterogeneity of the path without a parametric model, providing an interface for
transfer to real markets, and the grouped ablation further shows that the contribution of the jump- and
volatility-related feature groups is larger than that of any single feature;
while supplementing parametric calibration with Greeks at the training strikes cannot replace measure
learning, still failing at short-maturity deep OTM ($-99.6\%$), because first-order information cannot
recover the curvature outside the training range (Section~\ref{sec:bs_greeks}).

Finally, a note on the $25$--$130\times$ multiples reported earlier when EMS was the main baseline. These
multiples reflect the degree of failure of the scalar assumption of EMS under a risk premium, and in that
sense are not a numerical artifact; but they should not be read as the magnitude of XiNet's advantage over
a reasonable baseline: with WMC as the main baseline, all methods are nearly exact within the interpolation
range and the order-of-magnitude gap accordingly disappears, and under strict held-out XiNet does not
consistently outperform WMC. Therefore this paper takes the path-dependent claim at the measure level as
the core metric, and uses the multiples only as a diagnostic for EMS.

\subsection{Future work}

The deep-out-of-the-money and support limitations on real data are rooted in the insufficient tail mass of
the semi-parametric physical prior; introducing volatility-filtered historical simulation or a
heavier-tailed physical prior, combined with bid/ask robustification and hedging-error evaluation, may
improve out-of-range pricing. On the side of path-identifying instruments, one can include infinitely
active Lévy settings such as CGMY-JRP in the comparison of path-dependent claims (which requires a
higher-fidelity sampler), and evaluate the incremental identifying power of VIX, variance swaps, and
multi-observation-day constraints as path-identifying instruments. On the theoretical side, the
Girsanov--Meyer theorem~\citep{protter_stochastic_2005} guarantees that a semimartingale remains a
semimartingale under an equivalent measure but with a changed decomposition, and one can further establish
a general algorithm that learns $\mathrm{d}\Q/\mathrm{d}\Prob$ from $\Prob$ for any semimartingale without
specifying a model family a priori.

\section*{Acknowledgments}
We are grateful to Yang He, Xinyu Hu, Yang Lyu, Zihua She and Bochang Yang for insightful discussions on this work, which greatly improved the paper.
\newpage
\bibliography{ref}

@article{huang2014modified,
  title   = {A modified empirical martingale simulation
             for financial derivative pricing},
  author  = {Huang, Shih-Feng},
  journal = {Communications in Statistics --- Theory and Methods},
  volume  = {43},
  pages   = {328--342},
  year    = {2014}
}

@article{delbaen1994general,
  title   = {A general version of the fundamental theorem
             of asset pricing},
  author  = {Delbaen, Freddy and Schachermayer, Walter},
  journal = {Mathematische Annalen},
  volume  = {300},
  pages   = {463--520},
  year    = {1994}
}

@article{heston1993closed,
  title   = {A closed-form solution for options with stochastic
             volatility with applications to bond and currency options},
  author  = {Heston, Steven L.},
  journal = {The Review of Financial Studies},
  volume  = {6},
  number  = {2},
  pages   = {327--343},
  year    = {1993}
}

@article{frittelli2000minimal,
  title   = {The minimal entropy martingale measure
             and the valuation problem in incomplete markets},
  author  = {Frittelli, Marco},
  journal = {Mathematical Finance},
  volume  = {10},
  number  = {1},
  pages   = {39--52},
  year    = {2000}
}

@article{avellaneda1998minimum,
  title   = {Minimum-relative-entropy calibration of asset-pricing models},
  author  = {Avellaneda, Marco and Friedman, Craig and Holmes, Richard
             and Samperi, Dominick},
  journal = {International Journal of Theoretical and Applied Finance},
  volume  = {1},
  number  = {4},
  pages   = {447--472},
  year    = {1998}
}

@article{breeden1978prices,
  title   = {Prices of state-contingent claims implicit in option prices},
  author  = {Breeden, Douglas T. and Litzenberger, Robert H.},
  journal = {Journal of Business},
  volume  = {51},
  number  = {4},
  pages   = {621--651},
  year    = {1978}
}

@article{elices2008conditions,
  title   = {Conditions for the absence of arbitrage in {Implied
             Volatility} models},
  author  = {Elices, Antonio and Gim{\'e}nez, Eduard},
  journal = {SSRN Working Paper},
  year    = {2008}
}

@article{schweizer1996approximation,
  title   = {Approximation pricing and the variance-optimal martingale measure},
  author  = {Schweizer, Martin},
  journal = {The Annals of Probability},
  volume  = {24},
  number  = {1},
  pages   = {206--236},
  year    = {1996}
}

@article{duan_empirical_1998,
	title = {Empirical Martingale Simulation for Asset Prices},
	volume = {44},
	issn = {0025-1909, 1526-5501},
	url = {https://pubsonline.informs.org/doi/10.1287/mnsc.44.9.1218},
	doi = {10.1287/mnsc.44.9.1218},
	pages = {1218--1233},
	number = {9},
	journaltitle = {Management Science},
	shortjournal = {Management Science},
	author = {Duan, Jin-Chuan and Simonato, Jean-Guy},
	urldate = {2026-01-14},
	date = {1998-09},
	langid = {english},
}

@article{protter_partial_2001,
	title = {A partial introduction to financial asset pricing theory},
	volume = {91},
	rights = {https://www.elsevier.com/tdm/userlicense/1.0/},
	issn = {03044149},
	url = {https://linkinghub.elsevier.com/retrieve/pii/S0304414900000648},
	doi = {10.1016/S0304-4149(00)00064-8},
	pages = {169--203},
	number = {2},
	journaltitle = {Stochastic Processes and their Applications},
	shortjournal = {Stochastic Processes and their Applications},
	author = {Protter, Philip},
	urldate = {2026-01-22},
	date = {2001-02},
	langid = {english},
}

@book{protter_stochastic_2005,
	location = {Berlin, Heidelberg},
	title = {Stochastic Integration and Differential Equations},
	volume = {21},
	rights = {http://www.springer.com/tdm},
	isbn = {978-3-642-05560-7 978-3-662-10061-5},
	url = {http://link.springer.com/10.1007/978-3-662-10061-5},
	doi = {10.1007/978-3-662-10061-5},
	series = {Stochastic Modelling and Applied Probability},
	publisher = {Springer Berlin Heidelberg},
	author = {Protter, Philip E.},
	urldate = {2026-01-22},
	date = {2005},
	langid = {english},
}

@article{huang_multi-asset_2026,
	title = {{MULTI}-{ASSET} {EMPIRICAL} {MARTINGALE} {PRICE} {ESTIMATORS} {FOR} {FINANCIAL} {DERIVATIVES}},
	author = {Huang, Shih-Feng and Ciou, Guan-Chih},
	date = {2026},
	langid = {english},
}

@article{yuan_strong_2009,
	title = {Strong consistency of the empirical martingale simulation option price estimator},
	volume = {25},
	rights = {http://www.springer.com/tdm},
	issn = {0168-9673, 1618-3932},
	url = {http://link.springer.com/10.1007/s10255-008-8801-7},
	doi = {10.1007/s10255-008-8801-7},
	pages = {355--368},
	number = {3},
	journaltitle = {Acta Mathematicae Applicatae Sinica, English Series},
	shortjournal = {Acta Math. Appl. Sin. Engl. Ser.},
	author = {Yuan, Zhu-shun and Chen, Ge-mai},
	urldate = {2026-02-16},
	date = {2009-07},
	langid = {english},
}

@article{merton1976option,
  title={Option pricing when underlying stock returns are discontinuous},
  author={Merton, Robert C},
  journal={Journal of financial economics},
  volume={3},
  number={1-2},
  pages={125--144},
  year={1976},
  publisher={Elsevier}
}

@article{carr2002fine,
  title   = {The Fine Structure of Asset Returns: An Empirical Investigation},
  author  = {Carr, Peter and Geman, H{\'e}lyette and Madan, Dilip B.
             and Yor, Marc},
  journal = {Journal of Business},
  volume  = {75},
  number  = {2},
  pages   = {305--332},
  year    = {2002}
}

@article{hawkes1971spectra,
  title   = {Spectra of Some Self-Exciting and Mutually Exciting Point Processes},
  author  = {Hawkes, Alan G.},
  journal = {Biometrika},
  volume  = {58},
  number  = {1},
  pages   = {83--90},
  year    = {1971}
}

\appendix

\section{Summary of experimental parameters}
\label{app:params}

\begin{table}[H]
  \centering
  \caption{Parameters of the Merton model-specific feature experiment (Section~\ref{sec:experiments})}
  \label{tab:params}
  \begin{tabular}{lll}
    \toprule
    Parameter & Symbol & Value \\
    \midrule
    Initial price & $S_0$ & 100 \\
    Risk-free rate & $r$ & 0.10 \\
    Physical drift & $\mu$ & 0.15 \\
    Diffusion volatility & $\sigma$ & 0.20 \\
    Jump intensity & $\lambda$ & 0.5 \\
    Jump mean (log) & $\mu_j$ & $-0.10$ \\
    Jump volatility (log) & $\delta$ & 0.20 \\
    \midrule
    Number of training paths & $n$ & 22{,}000 \\
    Number of test paths & $n_{\mathrm{test}}$ & 55{,}000 \\
    Number of repetitions & $n_{\mathrm{rep}}$ & 120 \\
    Number of epochs & $N_{\mathrm{epoch}}$ & 3{,}000 \\
    Return loss weights & $\lambda_{\mathrm{stock}},\lambda_{\mathrm{call}}$ & $5,\,15$ \\
    Optimizer & & AdamW \\
    Learning-rate schedule & & Cosine annealing \\
    \bottomrule
  \end{tabular}
\end{table}

\begin{table}[H]
  \centering
  \caption{Parameters of the risk-premium experiments (Section~\ref{sec:risk_premium}, comparison of P and
  Q parameters)}
  \label{tab:rp_params}
  \small
  \begin{tabular}{llll}
    \toprule
    Setting & Parameter & P-measure value & Q-measure value \\
    \midrule
    \multirow{2}{*}{HestonVRP} & $\kappa$ (mean-reversion speed) & 2.0 & 5.0 \\
                                & $\theta$ (long-run variance) & 0.04 & 0.016 \\
    \midrule
    \multirow{3}{*}{MertonJRP} & $\lambda$ (jump intensity) & 0.5 & 2.0 \\
                                & $\mu_j$ (jump log-mean) & $-0.10$ & $-0.20$ \\
                                & $\delta$ (jump log-volatility) & 0.20 & 0.25 \\
    \midrule
    \multirow{3}{*}{CGMY-JRP} & $C$ (activity level) & 0.5 & 0.8 \\
                               & $G$ (positive-jump decay rate) & 8.0 & 4.0 \\
                               & $M$ (negative-jump decay rate) & 12.0 & 18.0 \\
    \midrule
    \bottomrule
  \end{tabular}
\end{table}
\section{Structural limitations of empirical martingale simulation}
\label{sec:ems_limit}

The EMS proposed by \citet{duan_empirical_1998} addresses the following phenomenon: even when the
theoretical model satisfies the martingale property, a finite Monte Carlo sample almost surely violates the
martingale condition. Its recursive correction applies the same per-time-step multiplier
$S_0/Z_0(t_j,n)$ to all paths, where $Z_0(t_j,n)=\frac{1}{n}e^{-rt_j}\sum_i Z_i(t_j,n)$, thereby forcing
the discounted sample mean to equal the current price $\frac{1}{n}\sum_i e^{-rt_j}S_i^*(t_j,n)=S_0$. From
the Radon-Nikodym viewpoint, as $n\to\infty$ we have $Z_0\to S_0$ (law of large numbers), so
$\mathrm{d}\hat{\Q}^{\mathrm{EMS}}/\mathrm{d}\Q\to 1$---\textbf{the EMS correction vanishes as the sample
size grows, it does not constitute a $\Prob\to\Q$ measure change, and this paper therefore uses it only as
a finite-sample control rather than a competing baseline.}

\textbf{Failure mechanism.} The true $\xi=\mathrm{d}\Q/\mathrm{d}\Prob$ is path-dependent under a jump
diffusion: paths with different numbers of jumps $N(\omega)$ have different relative weights between
$\Prob$ and $\Q$. The scalar correction of EMS applies the same scaling to all paths and cannot
distinguish paths with $N(\omega)=0$ from those with $N(\omega)=3$, so it necessarily produces a
systematic bias when the true $\xi(\omega)$ varies significantly with $N(\omega)$.

\textbf{Experimental evidence.} Under Merton model Setting A ($\Prob$ and $\Q$ differ only in drift:
$\mu=0.15$ vs.\ $r=0.10$, same jump parameters), over the full grid (25 $(T,S_0/K)$ grid points) the RMSE
of EMS is about $0.0022$ and the per-grid-point bias is $\leq 0.3\%$ (only $-1.1\%$ at deep OTM for
$T=1$M), and Neural $\xi$ has a bias of $\pm1\%$ in the same setting---both are nearly exact, confirming
that in the \emph{drift-only} case the scalar correction already suffices. Conversely, when $\Prob$ and
$\Q$ genuinely differ in the jump distribution (a risk premium exists), $\xi(\omega)$ is a non-degenerate
function of the path and the scalar correction is in principle insufficient---corresponding to the
MertonJRP and other settings of Section~\ref{sec:risk_premium} in the main text (EMS can underestimate by
up to $-75\%$ at deep out-of-the-money).

\end{document}